\documentclass[12pt]{article}

\pdfoutput=24
\usepackage{amsmath}
\usepackage{amsthm}
\usepackage{esvect}
\usepackage[toc,page]{appendix}
\usepackage{leftidx}
\usepackage{color}
\usepackage{framed, color}
\usepackage{multirow}
\usepackage{pdfpages}
\usepackage{multicol}
\usepackage{wrapfig,lipsum,booktabs}

\usepackage[utf8]{inputenc}
\usepackage{mathtools,hyperref}
\hypersetup{
    colorlinks=true,
    linkcolor=cyan,
    filecolor=cyan,      
    urlcolor=red,
    citecolor=red,
}

\usepackage{cleveref}
\usepackage{commath}
\usepackage{enumitem}
\usepackage{amssymb}

\definecolor{mgreen}{RGB}{25,147,100}
\definecolor{shadecolor}{rgb}{1,.8,.1}
\definecolor{shadecolor2}{RGB}{245,237,0}
\definecolor{orange}{RGB}{255,137,20}
\definecolor{orange}{RGB}{245,37,100}

\usepackage{pgfplots}
\usetikzlibrary{patterns}
\usepackage{mdframed}
\usepackage{adjustbox}
\usepackage{tcolorbox}
\usepackage{tikz,ifthen,fp,calc}
\usepackage{caption}
\usepackage{subcaption}
\usetikzlibrary{plotmarks}
\usepackage{graphicx}

\theoremstyle{plain}
\newtheorem{theorem}{Theorem}[section]
\newtheorem{prop}{Proposition}[section]
\newtheorem{corr}{Corollary}[section]
\theoremstyle{definition}
\newtheorem{definition}{Definition}[section]

\theoremstyle{definition}
\newtheorem{remark}{Remark}[section]

\usepackage[english]{babel}
\usepackage{babel,blindtext}

\newtheorem{exmp}{Example}[section]
\usepackage{fullpage}
\usepackage{amsfonts}
\usepackage{lscape}
\usepackage{bbm}

\usepackage{todonotes}
\usepackage{cite}
\usepackage{verbatim}
\usepackage{bm}

\usepackage[margin=1in]{geometry}
\providecommand{\keywords}[1]{\textbf{\textit{Keywords:---}} #1}

\usepackage[T1]{fontenc}
\usepackage[utf8]{inputenc}
\usepackage{authblk}
\usepackage{cite}

\title{\textbf{An energy-based interaction model for population opinion dynamics with topic coupling}}

\author[1]{Hossein Noorazar \thanks{hnoorazar@math.wsu.edu}}
\author[1]{Matthew J. Sottile\thanks{msottile@math.wsu.edu}}
\author[1]{Kevin R. Vixie\thanks{vixie@speakeasy.net}}
\affil[1]{Department of Mathematics and Statistics, Washington State University}

\date{}

\providecommand{\keywords}[1]{\textbf{\textit{Keywords:}} #1}

\begin{document} 

\maketitle

\begin{abstract}
  We introduce a new, and quite general variational model for opinion
  dynamics based on pairwise interaction potentials and a range of
  opinion evolution protocols ranging from random interactions to
  global synchronous flows in the opinion state space.  The model
  supports the concept of topic ``coupling'', allowing opinions held
  by individuals to be changed via indirect interaction with others on
  different subjects.
  Interaction topology is governed
  by a graph that determines interactions. Our model, which is really
  a family of variational models, has, as special cases, many of the
  previously established models for the opinion dynamics.

  After introducing the model, we study the dynamics of the special
  case in which the potential is either a tent function or a constructed
  bell-like curve. We find that even in these relatively simple potential
  function examples
  there emerges interesting behavior. We also present results of
  preliminary numerical explorations of the behavior of the model to motivate
  questions that can be explored analytically.
\end{abstract}


\section{Introduction}

An \emph{Opinion Game} is a mathematical model of the evolution of the
opinions of agents within a population, from creation of an idea to
its diffusion through the population and finally the ultimate state of
the system.  Opinion games (or opinion dynamics models) have been
studied since the middle of the 20th century.  They often are used to
study the dynamics of generic opinions that flow within a population.

By creating a formal mathematical model we can study detailed
hypotheses about transmission mechanisms and the impact of population
structure on opinion formation. Clearly such models are idealized and
fail to capture all the nuances and fuzziness of individual human
thought, but they give at least a disciplined glimpse into an aspect
of how agents interact. A formal analytic approach also admits the
ability to establish theorems about opinion dynamics, allowing one to
derive conclusions with a rigorous chain of formal reasoning
supporting them. This differs from approaches that take a purely
simulation-based approach in which conclusions are derived from
statistical properties of the resultant data instead of the underlying
dynamical rules dictating interactions. By building a solid foundation
for the opinion dynamics we get connection between model structure and
statistical properties.

Models typically vary in the assumptions that they make about the
nature of individual interactions and the overall goal of the
population in the limit of an arbitrary number of interactions.  In
some cases, models are used to study specific classes of opinions that
are created and flow within a population (such as linguistic
games\cite{Baronchelli20101,kay2003, Baronchelli2010, Puglisi2008}).
In other cases, the model seeks to explore
equilibrium states in which not all agents necessarily agree, but
reach opinion states that no longer change\cite{FRENCH1956}. The
generalized version of these games can be split into two categories. A
binary (or more generally, discrete) opinion
model~\cite{Follmer,Ding2010,Galam1991,Sznajd2000,Dietrich2003,Martins2007,GalamDiscrete}
requires agents to negotiate an opinion on a topic that has two
possibilities and therefore, the final state of the system divides
agents into two distinct groups, see \cite{Stauffer2003},
\cite{Biswas2009}. The alternative is a system in which the opinion
space is continuous~\cite{Weisbuch2002, Fortunato, Deffuant2000,
  Hegselmann2002}. For example ``\textit{How good do you think the New
  York Times is as a media outlet?}'' does not have a binary answer --
you might strongly believe it is a strictly good or bad source of
information, or you may have a mixed opinion somewhere in between. An
opinion game model allows us to study the state reached after a set of
individuals interacts an arbitrary number of times. This final state
may be consensus or there may emerge stable sub-populations
(\emph{clusters}) in which members of each group reach consensus, but
difference in opinion exist across these groups. The concept of
opinion can be extended to higher dimensions, allowing for more than
two extreme positions to be taken on a topic (e.g., ``what is your
favorite color?'').

\subsection{Understanding social dynamics}
%\begin{itemize}
%\item {\color{red}{How}} to model pairwise interactions, phenomena such as attraction/repulsion, expertise, trust, etc...

The dynamics of an opinion game are dictated by the micro-dynamics
that occur in individual-to-individual interactions between agents
connected by some network topology. At the micro-scale we would
represent the value of each opinion held by each individual, and how
opinions of two individual can change and evolve via direct
interaction. An example would be the one in which both agents have the
same power over each other and therefore, after an interaction they
would learn equally from each other and move relative to their
interaction partner by the same amount either in the direction of
increased agreement or disagreement or do not change whatsoever. For
example, repulsion could occur when two agent's opinion are too far
apart, especially when they talk about sensitive topics such as
religion or politics.  In some situations we may wish to model
interactions where a difference of opinion between two individuals is
insufficiently large to cause either attractive or repulsive behavior
in which no change occurs. In yet another situation we may model
dynamics in which opinions change only when they are sufficiently
different in a repulsive fashion, but with no attractive behavior when
they are near by. We will see how these can be modeled in
section~\ref{sec:interaction-potentials} via an interaction rule that
seeks to minimize an interaction energy via \emph{potential functions}.

\subsubsection{Asymmetry of influence}

In the real world, it is frequently the case that agents do not have
equivalent influence over each other. Such asymmetric relationships
are commonplace: teacher/student, parent/child, expert/amateur, and so
on. This can be modeled by applying weights, or \emph{influence
  power}, to the relationship between agents. In this case the
movements of the agents are not the same. For instance, if two agents
are talking about chemistry and one is a chemist and the other is a
student, then we can add an influence factor to the interaction and
make the chemist to have more influence on the student (expert
power). The result of such weighted influence would cause the student
to experience significantly higher changes of opinion than the teacher
during an interaction (which is part of the learning process), while
allowing the teacher to admit some small level of change as a part of
their interaction to maintain consistency in what they teach while
allowing for some level of flexibility in response to their
environment.

In an extreme case, we can consider the situation in which one agent
is absolutely dominant and acts purely as a \emph{speaker} and the
other is acting purely as a \emph{hearer} such that the opinion is
transmitted unidirectionally. Consequently after an interaction,
speaker would not change her opinion and hearer is the only one who
changes. This allows for the modeling of information sources such as
news or propaganda outlets that act in a purely influencing role
within the population.  It is worth mentioning that influence of
agents over each other might change over time and evolve, we refer
enthusiastic reader to~\cite{Friedkin2011}.  Expert power is one
mechanism by which asymmetry is introduced into an opinion model and
is applied on the interaction \textit{between} two individuals.

\subsubsection{Tendency for individualism}

The dynamics that emerge often fail to model phenomena
that appear in real social systems due to decisions made solely by
individuals. We can add more ingredients such as adaptive noise to
agents' decisions to capture these decisions.  This noise can be a
function of the collective opinion within the population or cluster
that an individual finds themselves a member of.  For example,
individuals that find themselves in a population that all hold the
same opinions may have a desire to express their individuality by
making an individual internal decision to change their opinion
slightly to differentiate themselves.  Exactly when and how this
desire for individuality will emerge and balance with a desire to
conform and form coherent like-minded groups is question open to
investigation.

Different real world phenomena emerge due to this desire for
individualism such as fads and trends, as well as the spontaneous
emergence of new trends. In the fashion world, when a new style
becomes popular there is a period in which individuals adapt their
opinions about fashion towards this popular style.  Eventually when
this style becomes widespread and relatively commonplace, some
individuals suddenly express a desire to be unique and begin avoiding 
this style. The more popular the style becomes, the more
powerful would be tendency of peers to try something new and leave the
cluster. This force is referred to as a ``\textit{disintegrating
  force}'' or ``\textit{tendency for individualism}'' or
``\textit{tendency for uniqueness}''\cite{Mas2010}. Many formal models
of opinion dynamics introduced in the 20th century are based on the
assumption that such disintegrating forces are not present, which
yields interesting but relatively synthetic idealized opinion
dynamics and equilibria.

\subsubsection{Coupling of topics}

It is rarely the case that opinions on specific topics exist in
isolation.  An opinion about one topic frequently influences the
opinion an individual holds on other topics.  These range from the
trivial (e.g., a preference for musical style influencing ones like or
dislike of a given musician in a given genre), to those that reflect
complex social influences (e.g., an opinion about religious belief
influencing opinions about reproductive rights). Mathematically we can
treat opinions that change together as being \emph{coupled}, where
there exists some functional that defines how the state of two or more
coupled opinions co-evolve.  In systems such as linguistic games to
determine color terminology, coupling rarely has an impact, but to
model social topics that have overlap at some level, coupling is an
unavoidable aspect of capturing the complexity and nuance of human
behavior. There are some works \cite{Li,Fortunato,Touri} in which
several topics are being discussed in the system, but the topics are
independent and coupling is not present. A recent work~\cite{Parsegov2015} studies coupling of interdependent topics for the Friedkin-Johnsen (FJ) model. A dynamical system view is
presented in \cite{Touri} along with an upper bound for convergence
time of the Hegselmann-Krause model~\cite{Hegselmann2002,blondel}.

\subsubsection{Interconnection topology}

Geographic and social factors dictate the likelihood of any two
individuals interacting.  This can be encoded via an interaction
network.  Given a population of individuals, the likelihood of any
pair of individuals interacting is strongly influenced by their social
and geographic network.  Agents who are friends are more likely to
interact than agents who are not, as are agents who are nearby versus
those that are physically separated by a large distance.  In our
model, we impose a graph topology on all individuals who can interact
in which the connectivity of the graph dictates which pairwise
interactions are possible.  For example, a tightly knit community of a
small number of individuals likely can be treated as a fully connected
graph in which all individuals may interact with each other.  On the
other hand, a set of individuals in widely distributed cities who do
not interact with each other but read the same newspaper can be
modeled as a star graph.  Realistic interaction networks can be
derived from social networking sites such as Twitter or Facebook,
allowing for connectivity graphs that mimic interactions that are
encountered in the real world.

%%%
%%% note: this is a specific consequence of our model, and not even a very
%%%       interesting one since it is trivially achieved by selection of
%%%       influence weights such that propaganda nodes aren't influenced at
%%%       all.  If anything, I would mention this as an aside or observation
%%%       when weights are introduced, as well as potentially using directed
%%%       graph edges.
%%%

%We mentioned agents tend to talk to those who are closer to them,
%geographically or philosophically. There are some powerful members of
%society, e.g. in the media, which cross those limiting borders, and
%target a broad range of audience with different backgrounds in a wide
%range of places. No one can deny the role of media or propaganda
%outlets in shaping and directing public opinions. To get more
%realistic we can add nodes that play this role to the network. Media
%act in a unidirectional fashion in which they share their opinions,
%but do not update them when they interact with individuals.  Moreover,
%when considering the topology of the interaction network, these
%classes of individuals frequently try to take advantage of having high
%connectivity and influence agents the way they want.  Understanding
%the impact of media outlets as individuals that interact with others
%with varying degrees of influence or weight, yet do not update their
%own opinion (e.g., they act in a ``broadcast-only'' mode), is an
%interesting question when examining social dynamics in systems such as
%a social network that contains advertising and news outlets.

\subsection{Contributions}

The model we present in this paper makes a number of noteworthy contributions to the study of opinion dynamics building upon decades of related research (Detailed in Sec.~\ref{sec:relatedwork}).

\begin{itemize}

\item The treatment of opinion evolution as a result of minimization
  of a (potentially nonlinear) energy functional based on the
  difference of opinion on a topic between interacting
  individuals. (Sec.~\ref{sec:interaction-potentials})

\item Demonstration that this model is a generalization of a family of
  previously published models from the literature by embedding them within
  the framework that we present.  (Sec.~\ref{sec:embedding})

\item Coupling of topics based on kernel-based coupling functions, to
  model the influence of explicit communication on one or more topics
  on a broader set of opinions held on topics not communicated.
  (Sec.~\ref{Topic-coupling})

\item A study of the dynamical systems properties of our model to
  understand the existence of regions of the opinion space that
  represent basins of attraction, fixed points and
  sensitivity. (Sec.~\ref{sec:dynamics})

\end{itemize}

A set of results of computational experiments are provided
demonstrating some of the dynamics one can observe in these models in
addition to the analytical results that we derive.  We close by posing
a set of directions of potential future research based on our model
and our early explorations with it.

\section{Related research}
\label{sec:relatedwork}

Models of opinion dynamics have been studied over a long period of
time with various formal methods and goals.  In this section we
provide a short review of noteworthy works in this history and discuss
how our work is novel relative to prior work.  This review is not
intended to be exhaustive.

\subsection{French model: early formal modeling of opinion dynamics}

One of the earliest works in this area appears in
1956~\cite{FRENCH1956}, in which French considers agents as particles
in a physical system which can induce forces on each other referred to
as \emph{interagental power}.  Interagental power can be characterized
in five ways: attraction power, expert power, reward power, legitimate
power and coercive power. Let $v_1$ and $v_2$ be two agents in the
network, then attraction power of $v_1$ over $v_2$ is how much $v_2$
likes $v_1$. Agents tend to listen to those whom they like and
respect. Expert power is based on $v_2$'s understanding of $v_1$'s
knowledge.  The student/professor example is the case in which a
student believes their teacher knows more on a given subject, and
therefore is influenced by their teacher.  Reward power is based on
the right of an individual to offer tangible rewards of any kind to
somebody for doing what is wanted or expected of them. For example,
parent might give a gift to a child for getting a good GPA in a year,
or a boss can promote one of their employees. A negative example of
reward power could be taking away a child's toy for not performing well
in school. Coercive power is based on $v_1$'s ability to impose
sanctions on $v_2$, it uses the threat of force to gain compliance
from another. The idea is that someone is forced to do something that
he/she does not desire to do. For instance in a court room the judge
has utmost absolute power over a defendant, or a boss has power over
their employee. And finally legitimate power is based on how much $v_2$
believes $v_1$ has a right to prescribe their opinion. It comes from an
elected or appointed position of authority and is based on the social
norm which requires agents to be obedient to those who hold superior
positions social structure.

At a given time, each individual is forced by the cumulative force
applied by all of their neighbors and their own resistance to move
towards an equilibrium point where the resultant force is zero. The
result is linear model that determines behavior of the system and
individuals. In this model, the tendency of individuals is to move
towards a weighted mean of their opinions at each step.  French also
explores the role of connectivity between individuals by applying the
model on different types of directed graphs, where direction of an
edge denotes a unidirectional (``speaks to'') relation.

%\subsubsection{Individualism which is based on Durkheim's theory of social integration - updates all-at-once}
%In French model updates are done all at once, is not it?
%{\color{red}{in section 3.5.2 we want to refer to collapse of clusters, which is related to the individualism, it is the same paper, so, do we want to break it into two pieces? one here and one there?}}

\subsection{DeGroot's averaging model and social influence evolution}

The model described by DeGroot\cite{DeGroot1974} is a simplified
version of the model we describe in this paper.  The core concept of
this model is that for a population of $n$ individuals, their set of
opinions on one topic at any given point in time can be represented as
a vector $\mathbf{y}^{(t)} \in \mathbb{R}^n$.  An $n$-by-$n$
row-stochastic matrix of influence weights, $\mathbf{P}$, is provided
in which the element $\mathbf{P}_{ij}$ corresponds to the amount of
influence that individual $j$ has on individual $i$ when updating
their opinion.  The diagonal elements $\mathbf{P}_{ii}$ represent each
individuals opinion of their own degree of influence within the
population, which can be interpreted as their confidence in their own
opinion relative to their peers.  The update process is simply
matrix-vector multiplication:

\begin{equation}
  \label{eq:DegrootEq}
  \mathbf{y}^{(t+1)} = \mathbf{P} \mathbf{y}^{(t)}
\end{equation}

In this model, given a set of initial opinions the result will
converge to the single weighted average of the initial opinions.  As
shown in Figure~(\ref{fig:degrootevolution}), for a full graph the convergence
is very regular for all individuals.  When a random graph is used, the
system still ends at an average value, but the individual trajectory
to this final value is more interesting (not necessarily monotonic).
Our model subsumes this behavior, and supports behavior that DeGroot's
model cannot capture.  For example, DeGroot's model implies that all
individuals will seek consensus - no matter how far apart the starting
opinions are, they will come together.  Our model captures the case
where sufficiently different opinions on a topic may result in an
increase in disagreement (e.g., polarization of opinions) in which the
final outcome is not a single value but two subpopulations that reach
two distinct opposing opinions.  We illustrate this in our
experimental results shown in Sec.~\ref{experiments}. As is clear
from Eq.~(\ref{eq:DegrootEq}), each node updates their opinion based on
their own and their neighbors opinions. An extension model in which each
agent incorporates their initial opinion throughout all iterations
was developed by Friedkin and Johnsen~\cite{Friedkin1999,FriedkinBook1,FriedkinBook2}.

\begin{figure}[httb!]
  \centering
  \includegraphics[width=.35\linewidth]{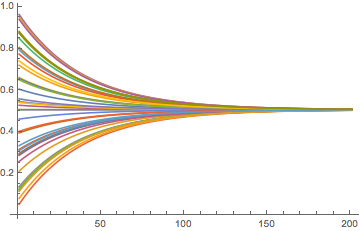}
\caption{Evolution of the DeGroot model from a set of initial opinions distributed over $[0,1]$ to consensus.}
\label{fig:degrootevolution}
\end{figure}

Moreover, an application of DeGroot's model, the DeGroot-Friedkin,
model has recently been studied with respect to the evolution of
social influence in work by Jia, et al.~\cite{jia2015}. In French's
model the influences are constant over
time. In~\cite{jia2015,Lorenz2005} authors model a new version of the
game where influences of agents evolve depending on the outcome of the
previous issue's result.

\subsection{The Bounded Confidence Model}

There are several different models for explaining and modeling the
reasons and dynamics of phenomena such as the reason why agents tend
to maintain their differences while they learn and become more alike
by interacting. Why does the process of becoming more alike stop
before a complete convergence to homogeneity of humans attributes?  In
1997, Axelrod injected the idea of homophily in \cite{Axelrod} to the
opinion game which is referred to as \textit{``bounded confidence"}
model (BCM).  In this model, individuals talk to agents who are
similar to them and consequently become even more similar and agents
whose opinions are too different do not interact. The pairwise
interaction was developed later on by Deffuant~\cite{Deffuant2000} and
a synchronized version is investigated by~\cite{Krause2000,
  dittmer2001consensus,blondel,Lorenz2007}. The influence of network
topology on BCM dynamics is investigated in~\cite{Weisbuch2004}.

In this model at each time step a random pair of neighbors is chosen
and the agents change their opinions if the difference of opinions are
less than a threshold. BCM succeeded in explaining the formation of
clusters, i.e. maintaining the differences between groups and not
coming to a total consensus, to some extent, but it was fragile and
sensitive.
%%%%%%%%%%%%
BCM generates opinion clustering in the context of discrete,
categorical opinions rather than continuous opinions. Moreover, it is
sensitive to ``interaction noise'' in the sense that even if with
small probability agents talk to others whose opinions are not
similar, i.e. talk to agents from different clusters, then we would
see the emergence of a consensus monoculture or polarized final
state. Therefore, BCM cannot fully explain coexistence of several
different clusters in the system.

This is where Durkheim's theory \cite{Durkheim} comes into
play. Durkheim argues that the integrating forces that bind
society together and binds an individual to others is that each
individual adopts a linear combination of others' opinions. M{\"a}s in
\cite{Mas2010} adopts Durkheim's theory and at a given time $t$,
updates opinion of individual $i$ using weighted average of the
difference between opinion of agent $i$ and the rest of
population. In this manner, the system would not be sensitive to
interaction noise and therefore tendency for individualism can be
added to the system so that it would not hurt coexistence of several
clusters in the society.

%%% added this (above) and removed the later subsection
Krause~\cite{Hegselmann2002} explains analytically and by simulations
what can be said about the final state of a game given its initial
profile for existing models in which opinion of each agent evolves as
weighted average of other agents in the network.  Different cases are
considered where the weights are held constant or evolve over
time. Negative ties, differentiation from and xenophobia toward those
who are different, and its polarization effect is explained in
~\cite{Flache2011,Proskurnikov,Altafini}. Moreover, Anahita
in~\cite{BCMDiverse} investigates the heterogeneous system, i.e the
model in which each node has their own confidence interval.

\subsection{Distributed behavior control and flocking models}

%{\color{blue}{1- Adaptive weights on edges, the closer the opinions, the more probability of interaction. \\
%2- Local Convergence can cause global polarization.}}
The set of opinions held by an individual on a set of topics can be
thought of as nothing more than the state within that individual with
respect to a set of variables.  The association with these states and
variables is largely an interpretation of this abstraction in order to
study a specific system.  Looking at other systems in which a set of
agents with state interact, we can find interesting related work not
originating in the study of opinion models in the context of
distributed control problems in robotics.  In~\cite{reif99} the
problem of distributed behavior control in autonomous robots is
investigated using a method similar to the energy potential model used
in our model.  This potential model for robotic control was introduced
even earlier by Khatib~\cite{khatib86} for obstacle avoidance
purposes.  In particular, the use of potential fields to model
interactions removes the need for any centralized controller or state.
Decision making can occur solely via pairwise interactions between
robots and updates within each individual robot based on these
pairwise interactions and knowledge about the interaction potential
field.  There is a striking similarity in the models of flocking
behavior (e.g., boids~\cite{reynolds86}) to those in opinion dynamics
in which a desire for population consensus (cohesion) is balanced with
individual attraction or repulsion depending on proximity.  Similar
noise injection mechanisms to the social tendency towards
individualism can also be found in flocking models to avoid having the
system fall into a steady state.

\section{Model Definition}
\begin{tcolorbox}
\begin{center}
\textbf{Notation used throughout the paper}
\end{center}
\begin{multicols}{2}
[
]
\begin{itemize}
\item $\psi$ potential function.
\item $\nabla \psi$ gradient of function $\psi$.
\item $\tau$ parameter for some potential functions like (tip of the) tent, (bend point of) BCM potential.
\item $\mathbb{O}$ opinion space.
\item $\mathbb{T}$ topic space, if topic space is finite, $|\mathbb{T}| = n$.
\item $o_i^{(t)}$ opinion of agent $i$ at time $t$ if there is only one topic in the system.
\item $o_i^{(t)}(s_k)$ opinion of agent $i$ at time $t$ about topic $s_k$.
\item $\dot{o}$ derivative of $o$ with respect to time.
\item $\Delta o_i^{(t+1)}(s_k) = o_i^{(t+1)}(s_k) - o_i^{(t)}(s_k) $.
\item  $d_{ij}^{(t)}(s) = o_i^{(t)}(s) - o_j^{(t)}(s)$. 
  
Difference between opinion of agents $i$ and $j$ about topic $s$ at a time $t$.
\columnbreak
\item $N(\mu, \sigma)$ normal distribution with mean $\mu$ and standard deviation $\sigma$.
\item $\alpha$ learning rate.
\item $\sigma(\mathbf{A})$ spectrum of matrix $\mathbf{A}$.
\item $\mathbf{J}$ Jacobian matrix.
\item $\mathbf{C}$ coupling matrix.
\item $c_{ij}$ inverse coupling strength of topic $i$ over topic $j$.
\item $G = (V,E)$ graph $G$ with set of nodes $V$ and set of edges $E$.
\item $n(i)$ set of neighbors of agent $i$.
\item $||.||$ Euclidean norm.
\item vectors and matrices will be bold letters.
\item $\mu_s$ mean of polarization counts
\item $\mu_p$ mean of stabilization time
\end{itemize}
\end{multicols}
\end{tcolorbox}

\subsection{Topic and opinion spaces}

First we will introduce definitions and associated notation that we will be using throughout this paper.  

\theoremstyle{definition}
\begin{definition}{}
The set of topics being negotiated in the model constitute the \textit{topic space}, denoted by $\mathbb{T}$.  The topic space is effectively an index set, and for $n$ topics we can assume $\mathbb{T} = \mathbb{Z}_n$, and for an uncountable set of
topics ordered on real line  we have $\mathbb{T} :=[0,L]$.

\end{definition}

\theoremstyle{definition}
\begin{definition}{}
The set of all possible (numerical) opinions, denoted by $\mathbb{O}$, 
is called \textit{opinion space}, and will be a subset of an m-dimensional hypercube
$ \mathcal{C}^m$, for some $m$.
\end{definition}

\begin{remark}[]
In this paper we focus on the case in which
$\mathbb{O} = \mathcal{C}^1= [0,1]$ and this leads to the
opinion state space $\mathbb{O}^N$ where $N$ is the population of the network.
\end{remark}

For a given topic we assign an opinion value to it from 
the opinion space. For example if the question is 
``How good is the New York Times as a news source?'' the opinion
 space would be $\mathbb{O} = [0,1]$.
 The interpretation of the endpoints is arbitrary, but we 
 will adopt the convention when discussing opinions 
 where zero corresponds to an absolutely negative 
 opinion about a topic and one corresponds to an 
 absolutely positive opinion.  For more complex topics 
 the opinion space is defined by either the $d-1$ 
 dimensional simplex where $d$ corresponds to 
 the number of extreme opinions that the topic 
 supports and an opinion would be any point 
 within the simplex, or d-dimensional hypercube 
 where each dimension is a topic's attribute. \\

For example, if we look at the the color space 
defined by three primary colors - Red, Green 
and Blue - which is called the color triangle, then 
the color space (opinion space) would be a 
triangle and favorite color of an individual is 
a single point lying within the triangle. 
This opinion space can be naturally embedded 
inside a hypercube of dimension at least 3.

On the other hand, we can take the RGB cube as our opinion space in
which each dimension corresponds to a hue with the intensity of each
hue varying between 0 and 1.  The favorite color of an individual is a
vector in $\mathbb{R}^3$ lying inside the cube.  The choice of
specific opinion space used to encode specific topics is a problem
dependent choice.

\theoremstyle{definition}
\begin{definition}{}
Opinion of an individual $i$ about topic $x$ at given time $t$ is 
denoted by $o_i^{(t)}(x)$. $o_i:  \mathbb{T} \rightarrow \mathcal{C}^m$ 
is a function of both time and topic.
\end{definition}
 We may drop each of the subscripts/superscripts when 
 the meaning is clear from the context.

\subsection{Interactions}

Individuals that interact to share and update their opinions are
connected via an \textit{interaction network} represented by a
graph. Let $G = ( V , E ) $ represent the network under consideration
where $V$ is the set of nodes and $E$ is the set of edges. The set of
all neighbors of node $i \in V$ is denoted by $n(i).$ In this
experimental results presented in this paper we only consider the fully
connected network.  One might consider different graphs such as star
graph, ring, random graph, etc. and study the effect of graph
measures, such as different centralities, on the opinion game.

Let $d_{ij}^{(t)}(x):= o_i^{(t)}(x) -o_j^{(t)}(x)$, 
$\Delta o_i^{(t+1)}(x) := o_i^{(t+1)}(x) - o_i^{(t)}(x) $ 
for arbitrary topic $x$, and $w_{ij}$ to be the amount 
of influence of node $i$ on node $j$. $w_{ji}$ and $w_{ij}$ 
are not necessarily the same as the influence that two individuals 
have on each other is not necessarily symmetric (as in the case of a teacher/student or parent/child relationship).

\theoremstyle{definition}
\begin{definition}{}
Let $G = ( V , E ) $  be a given network. A sub-graph of 
G, $\hat{G} = ( \hat{V} , \hat{E} ) $, where $\hat{V} \subset V$  and $\hat{E} \subset E$, 
is called a $\epsilon-$cluster if $\forall i \in \hat{V} \: \exists j \in \hat{V} s.t. \: |o_i - o_j | < \epsilon$.
\end{definition}

%%%%%%%%%%%%%%%%%%%%%%%%%%%%
%%%%%%%%%%%%%%%%%%%%%%%%%%%% Time Unit Section
%%%%%%%%%%%%%%%%%%%%%%%%%%%%
\subsection{Units of time}
\label{sec:time-units}

The model relates the state of opinions across the population 
to an abstract notion of time in order to study the evolution and 
dynamics of this opinion state.  As such we must carefully 
define what constitutes a unit of time and how this choice has
 an impact on the evolution of the model.  To start, we establish 
 that there exists a maximum speed at which information can be 
 propagated through the population in a single unit of time.  
 In this model where all individuals are connected by an interaction 
 network the minimum distance unit is a single edge in the network.  
 We adopt the convention then that the maximum distance that 
 information is allowed to flow is at most one edge from the origination 
 point of the information.

A single time unit therefore allows any valid set of interactions 
such that this restriction on information flow distance holds.  
Formally, evolution of the system in a time unit entails a subset 
of individuals $\hat{V} \subset V$ and subset of edges representing 
pairwise interactions $\hat{E} \subset E$. The restriction on information 
flow holds when for the subgraph of $G$, $\hat{G}=(\hat{V},\hat{E})$, 
there does not exist any path in $\hat{G}$ of length greater than 1.  
Additional restrictions may also be imposed if interactions are directed 
(e.g., in the case of a hearer and speaker where only the hearer updates their state), 
such as the in-degree of all vertices in $\hat{V}$ being restricted to at most 1.

\subsection{Interaction energy potentials}
\label{sec:interaction-potentials}

In our model the strength of an interaction is dictated by an
interaction \textit{potential function} $\psi$. While, in this paper
we focus on those potentials which depend only on differences of
opinions, our model does not require this in general. As a result, our
model has most of the previous models in opinion dynamics as special
cases. This is illustrated in~\ref{sec:embedding} where we show that the
Bounded Confidence model, the DeGroot model and the French model are
special cases of our model.

Examples of potential function are illustrated in
Fig.~(\ref{fig:example-potential}). The potential function determines
the degree to which individuals must react in adjusting their own
opinion depending upon the difference of their opinion
$d_{ij}^{(t)}(x)$ when interacting with another individual on topic
$x$.  The interaction rules that update opinions of individuals encode
this minimization by pushing individual opinions in the opposite
direction of the gradient of the interaction potential ($-\nabla
\psi$).
% \begin{remark}[]\label{NotJustDifferences}
% We established the model based on differences of opinions 
% which is reasonable, but it can be generalized to other models 
% which depends on opinions in the system in any way, not just differences of opinions . 
% Using potentials that do not depend on differences and letting 
% self interactions we will show the DeGroot model is a 
% special case of our model which does not uses the opinion differences.
% \end{remark}
For example in the BCM, which is a special case of our model, 
individuals would learn from each other when their difference 
is less than a threshold and they do not interact otherwise. 
This potential function is illustrated in Fig.~(\ref{fig:BCMpoten}) 
with a threshold of $\tau$. 
The tent function illustrated in Fig.~(\ref{fig:tentfig1}) represents 
the case where the individuals whose opinions difference is 
less than $\tau$ attract each other by becoming closer and 
repel by becoming further apart if their difference is more than 
$\tau$. In the case of the skewed flat top tent potential, shown 
in Fig.~(\ref{fig:skewtentfig1}), if the difference is between $\tau_l$ 
and $\tau_r$ then individuals are indifferent and will not change 
their position. The flat top tent potential is similar to that of the 
bounded confidence model, except that it not only forces 
individuals to come closer together when they are sufficiently 
similar but also causes them to repel when they are sufficiently different.
%%
%% Examples of potential function
%%
\begin{figure}[httb!]
\begin{subfigure}{.35\textwidth}
  \centering
  \includegraphics[width=1\linewidth]{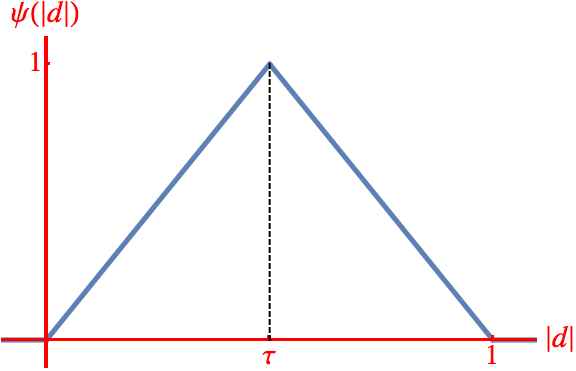}
  \caption{Tent Potential}
  \label{fig:tentfig1}
\end{subfigure}%
\begin{subfigure}{.35\textwidth}
  \centering
  \includegraphics[width=1\linewidth]{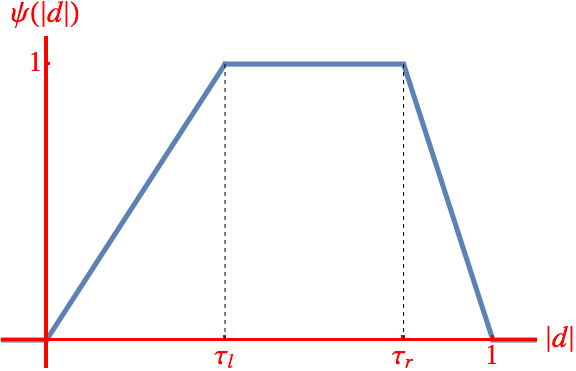}
  \caption{Skewed Flat Top Tent Potential}
  \label{fig:skewtentfig1}
\end{subfigure}%
\begin{subfigure}{.35\textwidth}
  \centering
  \includegraphics[width=1\linewidth]{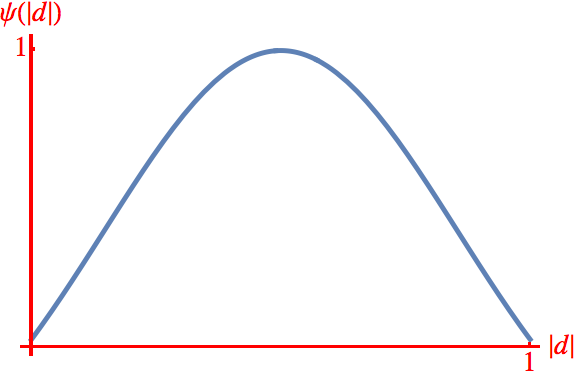}
  \caption{Gaussian Potential}
  \label{fig:sfig2}
\end{subfigure}
%% 
%% BCM potential Function
%%
\begin{subfigure}{.35\textwidth}
  \centering
  \includegraphics[width=1\linewidth]{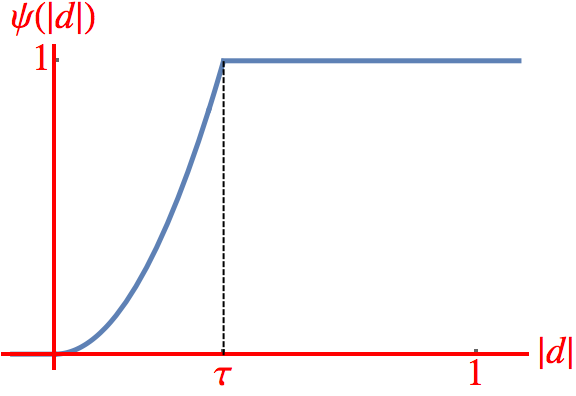}
  \caption{BCM Potential} 
   \label{fig:BCMpoten}
\end{subfigure}
\hspace{2.2in}
%% 
%% French potential Function
%%
\begin{subfigure}{.35\textwidth}
  \centering
  \includegraphics[width=1\linewidth]{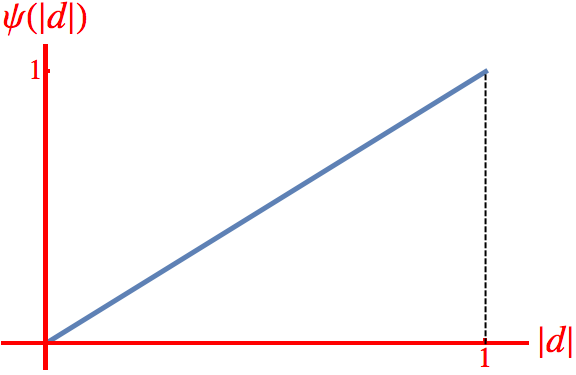}
  \label{fig:FrenchPoten}
\caption{A Simple Potential}
\end{subfigure}
\caption{Potential function examples}
\label{fig:example-potential}
\end{figure}

Given a pairwise update rule based on a potential function 
that dictates the attractive or repulsive relationship between 
two opinions, we must define an update rule for the entire population. 
Two options exist: one in which each step of the model involves 
only one pair of individuals, and one in which a maximal cohort 
is updated under a constraint dictated by the interaction network 
such that no individual makes more than one opinion update in 
a single step. In all cases the restrictions on information flow 
distance discussed in Sec.~\ref{sec:time-units} must be obeyed.

We will also discuss dynamics of the continuous game where 
the time is continuous and agents are continuously interacting.

\subsection{Individual and population update rules}

The model proceeds as a sequence of update steps in which 
one or more individuals update their opinion state based on a
 set of interactions.  Update rules can be either applied one-at-a-time, 
 in which random pairs of individuals interact, or in a concurrent 
 population wide fashion.  We start with the simplest case and build the model up. 
 The continuous game updates, of course, are done synchronously, like a physical system 
 with positively or negatively charged particles that are within reach
 of each other and exert force on one another. 

\subsubsection{One-at-a-time updates - Derivation of interaction update rule}

We let pairs of individuals interact and update their opinions. 
Suppose there is one topic, $x$, in the system and two individuals 
$i$ and $j$ are interacting at a time $t$. Let $\psi: \mathbb{R} \rightarrow \mathbb{R}$ 
be a bounded potential energy function which measures the energy in the opinion interaction:

\begin{equation}
Q_{ij}:= \psi(|d_{ij}|) = \psi(|o_i - o_j|)
\end{equation}
In order to reduce the energy of interaction, we have to move in 
the opposite direction of gradient of the energy function:

\begin{equation}
\frac{\Delta d_{ij}}{\Delta t} = -\alpha \frac{\partial}{\partial d}\psi(|d|)
\end{equation}
where $\alpha$ is a parameter that increases/decreases rate 
of change of opinions. We refer to $\alpha$ as the ``\textit{learning rate}''. 
This would give us the update rule as follows:

\begin{equation}
\frac{ \Delta d_{ij}} {\Delta t} = \frac{ d_{ij}^{(t+1)} - d_{ij}^{(t)}} {1} = -\alpha \psi'(|d_{ij}|) \: \frac{d_{ij}^{(t)}}{|d_{ij}^{(t)}|}
\end{equation}
consequently,
\begin{equation}
d_{ij}^{(t+1)} = d_{ij}^{(t)}  -\alpha \psi'(|d_{ij}|)  \: \frac{d_{ij}^{(t)}}{|d_{ij}^{(t)}|}
\end{equation}
If we want to nudge both agents' opinion by the same amount we must have:
\begin{equation}\label{eq:equalNudge}
\left\{
	\begin{array}{lr}
    o_i^{(t+1)} - o_i^{(t)} &= -\frac{\alpha}{2} \psi'(|d_{ij}^{(t)}|) \: \frac{d_{ij}^{(t)}}{|d_{ij}^{(t)}|} \\
	-(o_j^{(t+1)} - o_j^{(t)}) &= -\frac{\alpha}{2}  \psi'(|d_{ij}^{(t)}|) \: \frac{d_{ij}^{(t)}}{|d_{ij}^{(t)}|}
    \end{array}
\right.
\end{equation}
And therefore the update rule from time $t$ to $t+1$ is given by,
\begin{equation} \label{eq:updateNoWeights}
\left\{
  \begin{array}{lr}
    o_i^{(t+1)} &= o_i^{(t)}  -\frac{\alpha}{2} \: \psi'(|d_{ij}^{(t)}|) \: \frac{d_{ij}^{(t)}}{|d_{ij}^{(t)}|}	\vspace{.1in}\\
    o_j^{(t+1)} &= o_j^{(t)}  +\frac{\alpha}{2}  \psi'(|d_{ij}^{(t)}|) \: \frac{d_{ij}^{(t)}}{|d_{ij}^{(t)}|}
  \end{array}
\right.
\end{equation}
If the two agents have different influences on each other, then 
the opinions change by different amounts. We can dictate this 
fact by inserting the influence weights into Eq.\eqref{eq:updateNoWeights} as follows:
\begin{equation} \label{eq:WeightedUpdateRule}
\left\{
  \begin{array}{lr}
    o_i^{(t+1)}  &= o_i^{(t)}  - \frac{\alpha}{2} w_{ji} \:  \psi'(|d_{ij}^{(t)}|) \: \frac{d_{ij}^{(t)}}{|d_{ij}^{(t)}|}	\vspace{.1in}\\
    o_j^{(t+1)} &= o_j^{(t)}  +  \frac{\alpha}{2}  w_{ij} \: \psi'(|d_{ij}^{(t)}|) \: \frac{d_{ij}^{(t)}}{|d_{ij}^{(t)}|}
  \end{array}
\right.
\end{equation}
Note that $o_i^{(t)}(x)$ is function of both time and the topic.  So,
if there are finite number of topics in the system, then the total
energy between the two given individuals $i$ and $j$ would be the sum
of all energies of opinion interactions.  The tendency for
individualism could be added to the model as well at this point. We
will use the noise model introduced in \cite{Mas2010}.  By adding this
term to Eq.~\eqref{eq:WeightedUpdateRule}, we get:

\begin{equation}
\left\{
  \begin{array}{lr}
    o_i^{(t+1)}  &= o_i^{(t)}  - \frac{\alpha}{2} w_{ji} \:  \psi'(|d_{ij}^{(t)}|) \; \frac{d_{ij}^{(t)}}{|d_{ij}^{(t)}|} + \xi_i(t)	\vspace{.1in}\\
    o_j^{(t+1)} &= o_j^{(t)}  +  \frac{\alpha}{2}  w_{ij} \: \psi'(|d_{ij}^{(t)}|) \; \frac{d_{ij}^{(t)}}{|d_{ij}^{(t)}|}+\xi_j(t)
  \end{array}
\right.
\end{equation}
where $\xi_i(t)$ is a random (sample) value from a distribution $\xi_i(t) \sim N(0,\sigma_i(t))$ associated with individual $i$ at time $t$ with mean zero and variance defined by:

\begin{equation}
\label{eq:adaptive-noise}
\sigma_i(t) = \frac{s}{e-1} \Big ( - |N(i)| + \sum_{j \in N(i)} e^{1-|d_{ij}(t)|} \Big )
\end{equation}

where the parameter $s$ is used to manipulate the strength 
of disintegrating forces.  This tendency for uniqueness increases 
when the $d_{ij}$'s are small, i.e. there is high uniformity.

If $o_i^{(t+1)}$ becomes more than 1 or less than zero we 
should set it to 1 or zero, respectively.  This can be achieved 
by applying a clamping function:

\begin{equation*}
clamp(x) = 
\begin{cases} 0 & \text{if $x < 0$,} \\
x &\text{if $0 \leq x \leq 1$,} \\
1 &\text{if $x > 1$.}
\end{cases}
\end{equation*}

\subsubsection{ All neighbor interaction }

At this point we have derived a rule for updating opinions after a
single interaction.  Now we can use it for the model in which at each
time step, each agent updates their opinion according to all of their
neighbors' opinion, i.e. at each time step a agent would have
interaction with all of their neighbors, and their opinion in time $t+1$
depends on all of her neighbor's opinion in time $t$. Let the node $i$
have $k$ neighbors, then the updating rule is given by:

\begin{equation}\label{eq:synchronizedUpdate}
o_i^{(t+1)} = o_i^{(t)} + \sum_{j=1}^k -\frac{\alpha}{2} \: \psi'(|d_{ij}|)\; \frac{d_{ij}}{|d_{ij}|}
\end{equation}

The influence weights can be added accordingly. The influence 
weights that individual $i$ assigns to their neighbors and themselves 
(resistance force) has to add up to 1, so the weight matrix $w$ is
stochastic:

$$\sum_{j \in N(i)} w_{ji} = 1$$

This is the simplest case in which there is only one potential 
function for all interaction. We can also have a potential function
 for each pair or each agent could have her own potential function.

\subsubsection{Bidirectional vs unidirectional information exchange}

In the previous sections we assume that information flow 
is bi-directional in an interaction such that both individuals 
share their opinion state and both update it based on what 
they receive.  It is also possible to model interactions as 
uni-directional such that there is a well defined speaker 
that shares information to a receiving hearer but not vice 
versa.  This distinction can be encoded in the network
through the use of directed versus undirected edges.  
It must be noted that a single undirected edge between 
individuals $i$ and $j$ is not equivalent to two directed 
vertices in opposing directions ($i$ to $j$ and $j$ to $i$), 
as each edge corresponds to a single communication event.  

A similar effect can be accomplished with undirected graphs through
the use of zero weights.  An individual $i$ that influences others will
have weights $w_{i*}$ that are not necessarily zero, but if $i$ is not
influenced at all by others then $\forall j \neq i$, $w_{ji} = 0$.
The use of weights in this manner extends to graphless models in which
connectivity is based on spatial proximity where no notion of
directedness or edge connectivity exists.  We do not cover this mechanism
for individual connectivity in this paper.

% The use of directed edged can be used to cleanly integrate 
% propaganda or media individuals into the model.  
% In this case, a propaganda individual will have only 
% outward reaching edges and in-degree of zero.

\subsubsection{ Embedding other models into our model}
\label{sec:embedding}

To show BCM is special case of our model, let the potential function be
\begin{equation}
\psi(x) = \begin{cases} 
      x^2 & 0 \le x\leq \tau \\
      \tau^2 & \tau < x\leq 1 \\
   \end{cases}
\end{equation}
Then we have 
\begin{equation}
\psi'(x) = \begin{cases} 
      2x & 0 \le x\leq \tau \\
      0 & \tau < x\leq 1 \\
   \end{cases}
\end{equation}
Hence, using \eqref{eq:updateNoWeights} if the two agents are 
within confidence interval of each other we get:

\begin{equation} \label{eq:BCMspecial}
\left\{
  \begin{array}{lr}
    o_i^{(t+1)}  &= o_i^{(t)}  -\frac{\alpha}{2} \: (2|d_{ij}|) \: \frac{d_{ij}^{(t)}}{|d_{ij}^{(t)}|} =  o_i^{(t)} - \alpha \: d_{ij}^{(t)}	\vspace{.1in}\\
    o_j^{(t+1)} &= o_j^{(t)}  +\frac{\alpha}{2}  (2|d_{ij}|) \: \frac{d_{ij}^{(t)}}{|d_{ij}^{(t)}|} =  o_j^{(t)}  + \alpha \: d_{ij}^{(t)}
  \end{array}
\right.
\end{equation}
which is identical to the equation in section 2.1 of~\cite{Deffuant2000}. 
This is the pairwise interaction and it is easy to build the synchronized version.\\

We saw in the DeGroot model we have $\mathbf{y}^{(t+1)} = \mathbf{P} \mathbf{y^{(t)}}$. 
The following conditions will show the DeGroot model is a special 
case of our model. Consider synchronized version of the game with 
the following conditions. Define potentials 

\begin{equation}
\psi_{ij} = \psi(o_i,o_j) = \begin{cases} 
                       - p_{ij} o_i o_j &  j \ne i \\
      -\frac{1}{2} (p_{ii}-1) o_i^2 &  j = i \\
   \end{cases}
\end{equation}
then,
\begin{equation}
\frac{\partial}{o_i} \psi_{ij}  = \begin{cases} 
       -p_{ij} o_j      &  j \ne i \\
       -(p_{ii}-1) o_i &  j = i \\
   \end{cases}
\end{equation}

Now, the DeGroot model says $\mathbf{y}^{(t+1)} = \mathbf{P} \mathbf{y}^{(t)}$  
which is the same as \linebreak $\mathbf{y}^{(t+1)} - \mathbf{y}^{(t)} = (\mathbf{P}-\mathbf{I}_N) \mathbf{y}^{(t)}$ 
where $\mathbf{I}_N$ is identity matrix of size $N$. Since we are letting 
self interaction, defined by $\psi_{ii}$ for individual $i$ we have the 
following, whose right hand side will still give us a gradient descent 
method: (Let $\alpha = 2$)

\[
\begin{array}{lllllll}
o_i^{(t+1)} - o_i^{(t)}  &=& - \psi'_{i1}(o_i^{(t)},o_1^{(t)}) -  \psi'_{i2}(o_i^{(t)},o_2^{(t)}) - \cdots - \psi'_{ii}(o_i^{(t)},o_i^{(t)}) - \cdots  - \psi'_{iN}(o_i^{(t)},o_N^{(t)}) \\
             &=& p_{i1}o_1^{(t)} + p_{i2}o_2^{(t)} + \cdots + ( p_{ii} - 1 ) o_i^{(t)} + \cdots  p_{iN}o_N^{(t)} 
\end{array} 
\]
and therefore, 

\[
\begin{array}{lllllll}
o_i^{(t+1)}  &=& - \psi'_{i1}(o_i^{(t)},o_1^{(t)}) -  \psi'_{i2}(o_i^{(t)},o_2^{(t)}) - \cdots - \psi'_{ii}(o_i^{(t)},o_i^{(t)}) - \cdots  - \psi'_{iN}(o_i^{(t)},o_N^{(t)}) \\
             &=& p_{i1}o_1^{(t)} + p_{i2}o_2^{(t)} + \cdots + p_{ii} o_i^{(t)} + \cdots  p_{iN}o_N^{(t)} 
\end{array} 
\]
which is the $i^{th}$ row of the DeGroot model matrix.\\

The interesting dynamics of French model~\cite{FRENCH1956} which 
works with digraphs comes from the topology of the network . 
Let $i$ be a node with indegree $k$ and define $N(i) = \{i_1, i_2, \cdots , i_k\}$ 
be the neighbors of agent $i$ who have power over it. Define the potential functions

\begin{equation}
\psi_{ij} = \psi(o_i,o_{i_j}) = \begin{cases} 
      - \frac{o_i o_{i_j}}{k+1} &  1 \le i_j \le k \\
      \frac{k}{2(k+1)} o_i^2 &  i_j = i \\
   \end{cases}
\end{equation}
where $o_{i_j}^{(t)}$ is opinion of $j^{th}$ neighbor of $i$ 
which has power over $i$, i.e. there is an incoming edge 
from $i_j$ agent to $i$, and the case $i_j = i $ is the 
self-resistance element of French model, then we have
\begin{equation}
 \frac{\partial}{\partial o_i} \psi_{ij} = \begin{cases} 
      - \frac{ o_{i_j} } { k+1 } &  1 \le i_j \le k \\
      \frac{k}{(k+1)} o_i &  i_j = i \\
   \end{cases}
\end{equation}
and therefore, moving in the opposite of gradient of the potential function with synchronous update gives us the following step size
\[ o_i^{(t+1)} - o_i^{(t)}= \frac{1}{k+1} ( o_{i_1}^{(t)} + o_{i_2}^{(t)} + \cdots + o_{i_k}^{(t)} ) - \frac{k}{k+1} o_i^{(t)}\]
which yields the French update rule:
\[ o_i^{(t+1)} = \frac{1}{k+1} (  o_i^{(t)} +  o_{i_1}^{(t)} + o_{i_2}^{(t)} + \cdots + o_{i_k}^{(t)} )\]

\subsection{Topic coupling}
\label{Topic-coupling}
In the real world, topics are not independent of each other. For example, 
there might be good reasons that if someone changes their opinion 
about education then they would also change their mind about social
topics in which education is a factor.

\theoremstyle{definition}
\begin{definition}{}
Change of opinion about topic $s_k$ as a result of change of opinion about topic $s_\ell$ is called \emph{coupling}.
\end{definition}

The coupling is done in two ways. The first case is the discrete case
where there are a finite number of topics in the system. In a game the
two agents can talk about some of the topics and then the coupling is
performed after the interaction is complete.  In the second model
there are uncountable number of topics, at a given time step two given
individuals could talk about all topics at the same time to the extent
they please, i.e. they have the option of not talking about some of
the topics or revealing all information about a given topic in a game,
and that is captured by the so called \textit{conversation filter}
function. This generalized model consist of three steps.  In this
first phase, each opinion is changed only with the force enforced by
negotiation about the given topic and in the second phase the coupling
is done between topics within each individual prior to any other
interactions. In the coupling phase some opinions might be pushed out
of the opinion space they belong to, hence, in the third phase we
would push the out of place opinions back in place.
\subsubsection{Discrete coupling}

\theoremstyle{definition}
\begin{definition}{}
Given a topic space $\mathbb{T} = \{ s_1, s_2, \cdots, s_n \}$ and let
$\mathbf{C}$ be the \emph{coupling matrix} for the system. For two
given topics $s_{\ell}$ and $s_k$, we define the \textit{inverse
  coupling strength}(or \textit{inverse coupling coefficient} or
\textit{inverse correlation strength}), denoted by $c_{\ell k}$,
to be the parameter that determines how much the two topics are
influenced by each other. The coefficient $c_{\ell k}$ dictates how the
opinion $o(s_k)$ will change in response to a change in opinion
$o(s_{\ell})$ that occurs due to an interaction.
$c_{\ell k}$ does not equal $c_{k\ell}$ necessarily.
\end{definition}

A positive coupling coefficient would mean the two topics have direct
correlation and a negative coefficient means the two topics have an
inverse correlation. Topics are by definition completely correlated
with themselves, so $c_{s_is_i}=1$ for all topics $s_i$.

To derive the update rule for coupled topics, suppose that individual
$i$ has their opinion changed about a topic after an interaction. This
change is due to an external force (negotiation force) resulting from
the opinion held by the other individual on the subject discussed
during the interaction. Let $f:= 1 \cdot \Delta o_i(s_{\ell}) $ and
let topic $s_{\ell}$ to be related to topic $s_k$ with coefficient
$c_{\ell k}$. Then the same force would be applied to opinion of the
agent about subject $s_k$ and therefore:

\[f = 1 \cdot \Delta o_i (s_{\ell}) = c_{\ell k} \; \Delta o_i (s_{k})  \]
which gives us the update rule for the opinions that are correlated to the topic $s_{\ell}$: 

\[ o_i^{(t+1)}(s_k) =  o_i^{(t)}(s_k) + \frac{1}{c_{\ell k}} \; \Delta o_i^{(t)}(s_\ell). \]

\subsubsection{Continuous Kernel-based coupling}

In the previous section we modeled a discrete version of coupling where 
there are finite number of topics in the system and each pair of topics are 
coupled with their specific coupling coefficient. Here we would model the generalized 
continuous topic space where the coupling coefficients are determined 
via a two variable function and each pair of agents have the option of 
revealing all of their information about all topics in a single game. 
So, in a single time step, the two agents have the choice of talking 
about as many topics as they please, to whatever extent they prefer. 
For example, two agents can talk about their favorite food and favorite 
chef in an interaction, but not their favorite president. 
This is implemented via a conversation \textit{filter} $f_{ij}(x,t) \in [0,1]$ 
which can evolve over time. $f_{ij}(x,t)=0$ means they would not 
talk about topic $x$, $f_{ij}(x,t)=1$ would mean they reveal all of 
their information about the given topic. It will not effect the opinions that are revealed, but it  determines the weight that each revealed opinion contributes to the interaction energy.\\

Let $\mathbb{T} := [0,L]$ be the topic space and $\mathcal{S}$ be the set 
of measurable functions \linebreak $f : [0,L] \rightarrow [0,1]$. 
Suppose further that $\psi: \mathbb{R} \rightarrow \mathbb{R}$ 
is bounded. Then $o_i \in \mathcal{S}$, $\forall i \in V$. 
We define the interaction energy between nodes $i$ and $j$ over topic space by:

\[ Q_{ij} = \int_0^L \psi(||o_i^{(t)}(x) - o_j^{(t)}(x) ||) \; f_{ij}^{(t)}(x)dx \]
Minimizing this energy would lead us to an update rule for the agents at topic $x$. 
 
After the conversation, each individual has updated her opinion 
about any given topic via interaction and then it is time to do the 
coupling in the topic network. In order to do this, we measure the 
energy between two topics for node $i$ as follows 
(we can have different potential energy $\phi$ for every pair of 
topics $x$ and $y$, hence we use the notation $\phi(||o_i(x) - o_i(y)||,x,y)$):

\[ U_i[x,y] = \int_0^L  \int_0^L \phi(||o_i(x) - o_i(y) ||, x, y) \; k(x,y) dx dy \]

Where the kernel $k(x,y)$ determines the connection strength between topics $x$ and $y$.\\
To obtain the update rule for the two individuals we take the variational derivative of functional $U$.\\

Let  \[  U_i[x,y] = \int_0^L  \int_0^L \phi(||o_i(x) - o_i(y) ||, x, y) \; k(x,y) dx dy \]

In ordinary calculus $\Delta f$ is computed (as a function of $x$) when $x$ changes 
by $\Delta x$, here, $\delta U$ is  computed when $o_i$ changes by $\delta o_i$.
{\small
\begin{eqnarray*}
\frac{d}{d \epsilon} U_i[o(x)\hspace{-1.75em}&& + \epsilon h(x) , o(y) ]\Big|_{\epsilon=0} \\ &&=  \left. \bigg[\int_0^L  \int_0^L \partial_1 \phi(||o(x) +\epsilon h(x) - o(y) ||, x, y) \Big( \frac{o(x)+\epsilon h(x) - o(y)}{ || o(x)+\epsilon h(x) - o(y)|| }\cdot \: h(x)\Big) k(x,y) dx dy\bigg] \right|_{\epsilon=0}\\ 
&&=\int_0^L  \int_0^L \frac{\partial}{\partial o_i}\phi(|| o(x) - o(y) ||, x, y) \; \Big( \frac{o(x) - o(y)}{ || o(x) - o(y) || } \cdot h(x) \Big)  \; k(x,y) dx dy\\ 
&&=\int_0^L \bigg( \int_0^L \partial_1\phi(||o(x) - o(y)||, x, y) \frac{o(x) - o(y)}{ || o(x) - o(y) || } \; k(x,y) dy \bigg)\cdot h(x) dx\\
\end{eqnarray*}
}
Therefore, the derivative which gives the updating rule would be:

\[ \frac{\delta U}{\delta o_i(x)}  =  \int_0^L \partial_1 \phi(||o_i(x) - o_i(y) ||, x, y) \:  \frac{o_i(x) - o_i(y)}{ || o_i(x) - o_i(y) || } \; k(x,y) dy \]
which yields to:
\begin{equation} \label{eq:Coupiling-Topics}
o_i^{(t+1)}(x) = o_i^{(t)}(x) - \alpha \int_0^L \partial_1 \phi(||o_i^{(t)}(x) - o_i^{(t)}(y) ||, x, y) \Big( \frac{o_i^{(t)}(x) - o_i^{(t)}(y) }{ || o_i^{(t)}(x) - o_i^{(t)}(y) || }  \Big) \; k(x,y) \: dy
\end{equation}

For the reason explained below, we break the time step into two 
parts, $(t,\: t+1/2)$ and \linebreak $(t+1/2, \: t+1)$. Therefore, the 
left side of Eq.~\eqref{eq:Coupiling-Topics} is modified to 
$o_i^{(\tilde{t})}(x)$, where $\tilde{t} = t + 1/2$.

\section{Dynamics of the opinion game}\label{sec:dynamics}

The family of opinion game models we have introduced is rich and
flexible.  While the special case of symmetric interactions yields
gradient flows, the full range of this model is much broader. But even
in the gradient flow case, dynamics need not be simple: a sufficiently
complex potential can lead to very complex dynamical behavior, at
least computationally.

In this section we begin the exploration of model dynamics. Much of
the section focuses on the behavior in the continuous, synchronous
limit for systems in which the interaction potentials are symmetric
and identical. But even in these simpler cases we find behavior that
is not completely trivial.

We first study the behavior of a two agent system, each agent having a
different interaction potential.  This is simple enough that the
behavior can be understood in detail.

Moving to three agent systems in which all potentials are the same
symmetric tent function, we find that in addition to consensus
clusters, there are entire neighborhoods of fixed points. For a
deterministic interaction order, the discrete, asynchronous game can
see subsets of these regions as fixed regions in which the system
cycles endlessly. Even when the interaction order is random, these
regions of fixed points can lead to slow drifting random walks in
opinion space. The results from these studies are immediately relevant
for systems of $N = 3k$ agents for any $k$.

We then turn to potentials that are no longer tent potentials to find
non-trivial stable fixed points. These bell shaped potentials show us
that even in the case of gradient systems (which is what we have when
all potentials are the same) there can be interesting equilibrium
states.

While our explorations in this section are admittedly just the
beginning of a more comprehensive study of the model family, what we
show here already suggests some richness in behavior that encourages
us to push ahead with the other, more complex members of this family
of models.

%%%%%%%%%%%%%%
%%%%%%%%%%%%%% Dynamics of 3 Agents
%%%%%%%%%%%%%%
\subsection{Two Agent Systems}
\label{TwoAgentsSystem}

In this section, we consider the dynamics of a two agent system:

\begin{equation}
\label{eq:DifferenceVectorField}
	\begin{array}{llr}
  \dot{o}_1 &=& f_1(o_1,o_2) = -\frac{\partial\psi_1(|o_1-o_2|)}{\partial o_1}\\
  \dot{o}_2 &=& f_2(o_1,o_2) = -\frac{\partial\psi_2(|o_1-o_2|)}{\partial o_2}
    \end{array}
\end{equation}

The situation is simple enough that we can understand everything in
detail. Two observations enable us to fairly easily unravel the
dynamics of any two agent system:
\begin{enumerate}
\item the vector field in the 2-dimensional opinion state space,
  $\{(o_1,o_2) \;|\; o_1,o_2\in [0,1] )\}$, is constant along $o_1 -
  o_2 = c$ lines.
\item the dynamics are completely determined by the positions of the
  maxima and minima of the potential functions.
\end{enumerate}
To illustrate this, we look at a couple of systems in which each
individual's interaction potential is a simple tent function, though
with different peak positions. We are able to give a complete picture
of the simple dynamics in each case. 

What about the case when the potentials are not functions only of the
difference $|o_1-o_2|$? I.e. when
\begin{equation}
\label{eq:VectorField}
	\begin{array}{llr}
  \dot{o}_1 &=& f_1(o_1,o_2) =  -\frac{\partial\psi_1(o_1,o_2)}{\partial o_1}\\
  \dot{o}_2 &=& f_2(o_1,o_2) =  -\frac{\partial\psi_2(o_1,o_2)}{\partial o_2}
    \end{array}
\end{equation}
While this question is mostly left to future papers, we close the
section with a simple sensitivity result for the case of general
potentials.

\subsubsection{Two tents}
\label{sec:twotents}

In Fig.~(\ref{fig:EqualTipsSubFigure}) we see the case of two tent
potentials with their tips at $\tau_1$ and $\tau_2$ where $\tau_1 \leq
\tau_2$, and $1-\tau_2 = \tau_1$. In this simple case, the 2
dimensional state space is divided into regions in which the
vector field generated by the partial derivatives of the potentials are
constant.  There are three cases:
\begin{enumerate}
\item When $ o_1 - o_2 < \tau_1$, the agent's opinions move
towards each other and the slope of vectors is $\frac{\Delta
  o_2}{\Delta o_1} = - \frac{\tau_1}{\tau_2}$.
\item If $ \tau_1 < o_1 -
o_2 < \tau_2$, then, agent 1 moves away with agent two in pursuit,
$\frac{\Delta o_2}{\Delta o_1} = \frac{1 - \tau_1}{\tau_2} =
1$.
\item Finally, if $ o_1 - o_2 > \tau_2$ they repulse each other
towards polarization.
\end{enumerate}
We denote the region in which the opinions are
attracted towards each other by A, the region where one opinion chases
the other we denote by C, and R is used to label the region where they
repulse each other.

%%%%%%%%%
%%%%%%%%% Two Agents Vector Field Figure
%%%%%%%%%
\begin{figure}[httb!]
\hspace{.6in}
%%%%% Equal Tip Figrue
\begin{subfigure}{.45\textwidth}
  \centering
  \includegraphics[width=1\linewidth]{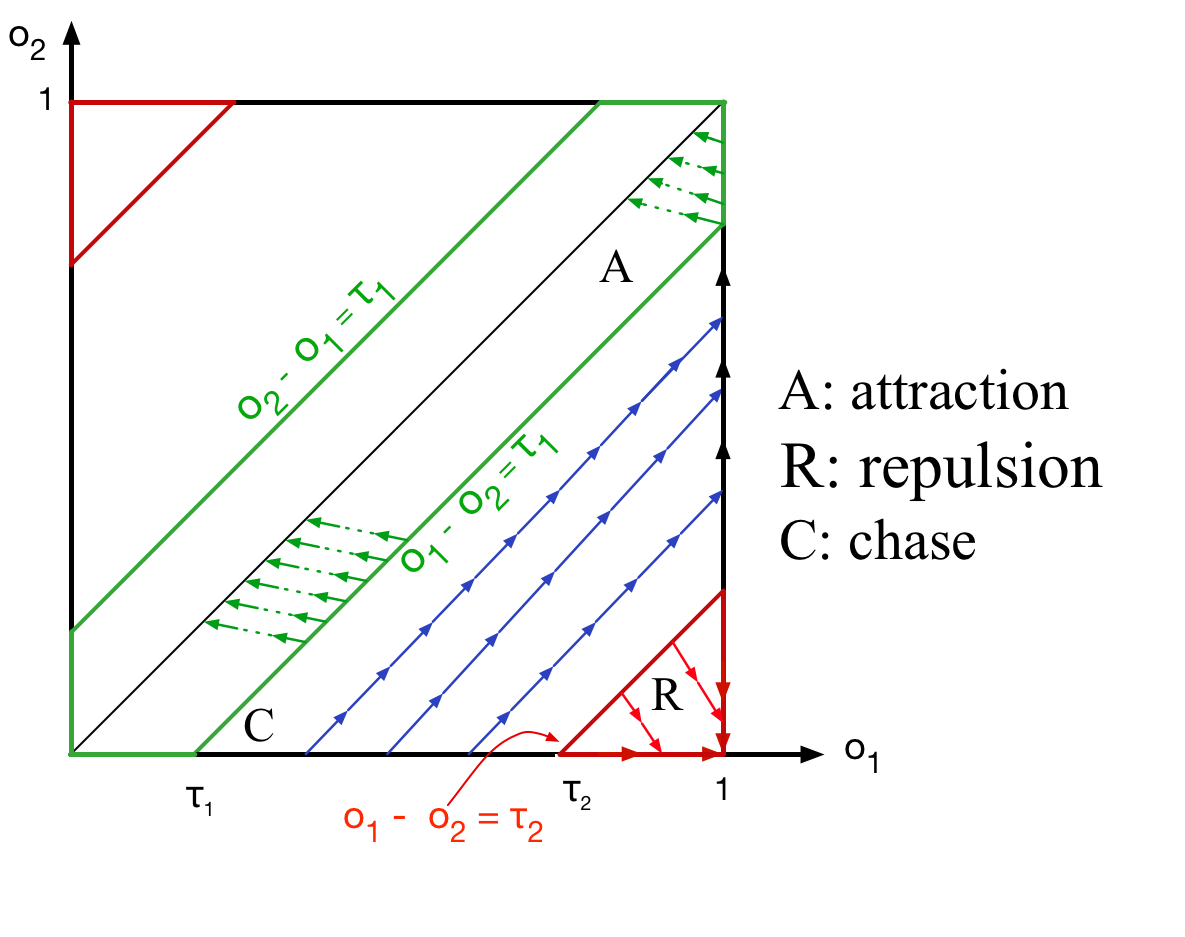}
  \caption{ $1 - \tau_2 = \tau_1$}
  \label{fig:EqualTipsSubFigure}
\end{subfigure}%
\hspace{.2in}
%%%%%%%% Chase Area Figure
\begin{subfigure}{.45\textwidth}
  \centering
  \includegraphics[width=1\linewidth]{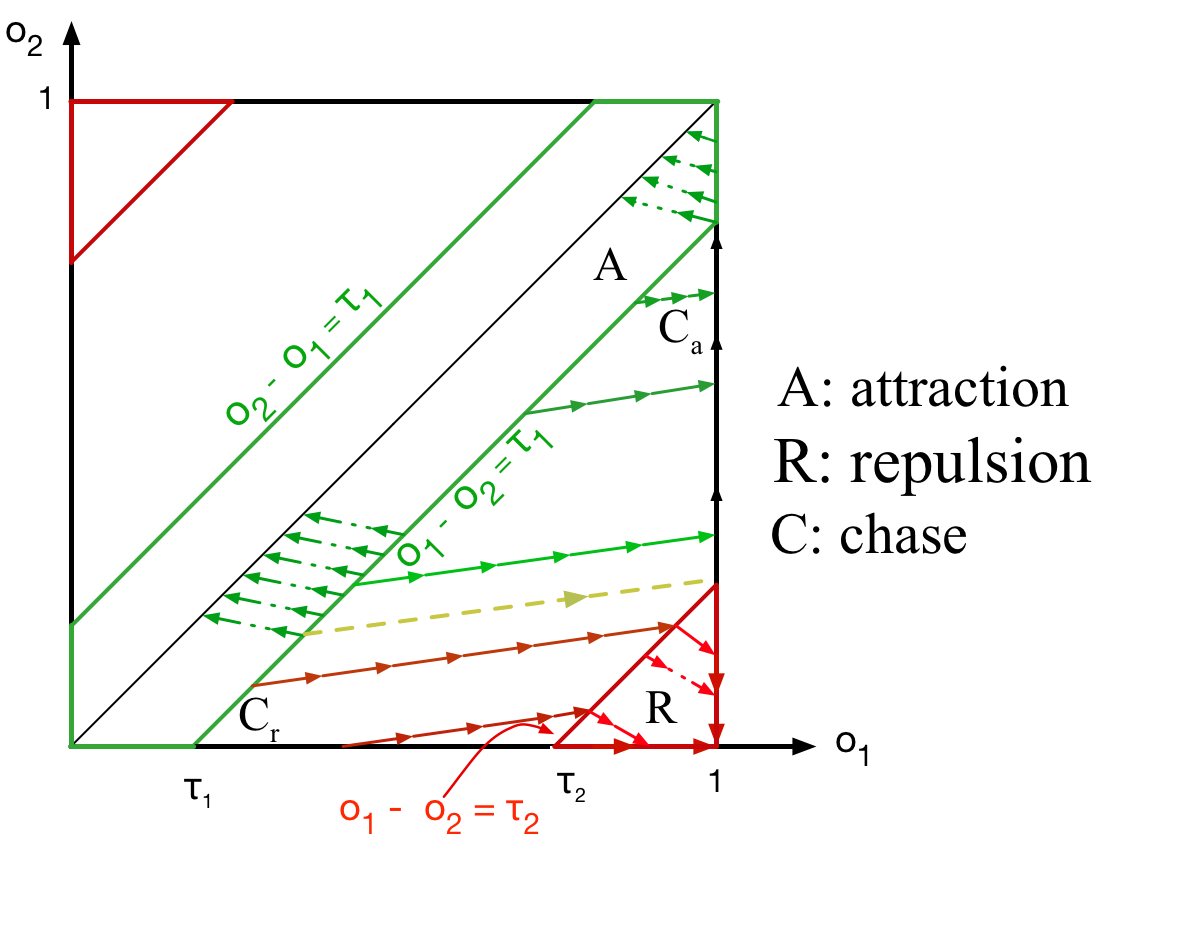}
  \caption{$1 - \tau_2 < \tau_1$}
  \label{fig:ChaseAreaSubFigure}
\end{subfigure}%
%%%%%%
\caption{Two agents vector fields}
\label{fig:TwoAgentsVectorFields}
\end{figure}

In the Fig.~(\ref{fig:ChaseAreaSubFigure}) we show the case of two tent
potentials with their tips at $\tau_1$ and $\tau_2$ where $\tau_1 \leq
\tau_2$ and $1-\tau_2<\tau_1$. Again we use A, C, and R for the
attraction, chase and repulsion regions, but now the chase region is
divided into 2 sub-regions, $\text{C}_a$ and $\text{C}_r$, depending on
whether the opinion state eventually enters the attraction region or
the repulsion region.

\begin{enumerate}
\item In region A, the slope of vectors is $m = \Delta o_2 / \Delta
  o_1 = -\tau_1 / \tau_2 < -1$, directed towards consensus.  All
  points on the line segment that connects $p = (a,a)$ to $q=
  (a+\cfrac{\tau_1\tau_2}{\tau_1+\tau_2}, \: a -
  \cfrac{\tau_1^2}{\tau_1+\tau_2})$ flow towards to $p = (a, \: a)$.
\item We subdivide region C into $\text{C}_a$ and $\text{C}_r$ regions:
  \begin{enumerate}
  \item The line $o_2 = \frac{1-\tau_1}{\tau_1}o_1 + \tau_2$ divides C
    into two regions.
  \item Above that line the chase eventually ends with in region A.
  \item Below that line the chase eventually ends in region R.
  \end{enumerate}
\item In the region R $o_1 - o_2 > \tau_2$, the opinions repulse and
  $o_1$ converges to $1$ and $o_2$ to $0$.
\end{enumerate}

\begin{remark}
  We are ignoring the fact that the tent system has undefined gradient
  at the tent tips, though it is not hard to deduce the results if
  smoothed out versions are used instead.
\end{remark}

\subsubsection{Two Smooth Potentials}
\label{sec:twolip}

When $\psi_1 = \psi_2$ or when $\psi_1$ and $\psi_2$ are piecewise
linear, the trajectories in the opinion state space are piecewise
linear and the sensitivity (separation of orbits) is very simple to
understand: mostly orbits that start close, remain close to each
other. When there are two different potentials with non-constant
derivatives, e.g. two different Gaussian-like potentials, the vector
field becomes slightly more interesting, though it is still constant
on $o_1-o_2 = c$ lines. This again results in a system whose
sensitivity to initial conditions are quite straightforward to
understand. (The simple structure of the lines of unstable fixed
points gives us the structure of the regions in which small
perturbations result in large differences in final equilibrium
positions.)  When we relax the assumption that the potentials are
functions of the difference $o_1 - o_2$ alone, we can get more
interesting behavior. In this section, we close with an example of the
kind of result we can get when, instead of assuming $\psi_1$ and
$\psi_2$ are functions of $|o_1-o_2|$ alone, we assume that $\psi_1$
and $\psi_2$ are smooth.

% The differential system defining the vector field is given by 
% (Let $\gamma = o_1 - o_2 \in [0,1]$):
% \begin{equation}
% \label{eq:VectorField}
% f = \begin{bmatrix}
%        f_1 \\[1.3em]
%        f_2
%      \end{bmatrix} = 
%      \begin{bmatrix}
%        \dot{o_1} \\[1.3em]
%        \dot{o_2}
%      \end{bmatrix} = 
%           \alpha/2 \begin{bmatrix}
%        -\psi'_1(|o_1 - o_2|)  \frac{ o_1 - o_2}{ | o_1 - o_2 |} \\[1.3em]
%         \psi'_2(|o_1 - o_2|)  \frac{ o_1 - o_2}{ | o_1 - o_2 |}
%      \end{bmatrix} = \alpha/2 \begin{bmatrix}
%        -\psi'_1(|\gamma|) \frac{\gamma} { |\gamma| } \\[1.3em]
%         \psi'_2(|\gamma|) \frac{\gamma} {| \gamma|}
%      \end{bmatrix}
% \end{equation}
% Having the vector field defined by differential system given by
% Eq.~(\ref{eq:VectorField}) we can study how perturbation of initial
% opinion can grow, we can think of this perturbation as initial input
% error.

\begin{prop}\label{driftOffProp}
  Let $G$ be a network with two agents with smooth, compactly
  supported potential functions, $\psi_1$ and $\psi_2$, generating a
  vector field, as in Eq.~(\ref{eq:VectorField}), with Lipschitz
  constant $L$. Let $p = (o_1^{(0)}, \: o_2^{(0)} )$ be the initial
  opinion state of the system, and $\hat{p} = (\hat{o}_1^{(0)}, \:
  \hat{o}_2^{(0)} ) $ be a perturbation of $p$. Then a bound on the
  difference between the states at time $T$ is given by \[||p^T -
  \hat{p}^T|| \leq || p^{(0)} - \hat{p} ^{(0)} || e^{TL^*}\] where
  $L^* \sim O(L)$.
\begin{proof}
  Let $\psi_1$ and $\psi_2$ be smooth, compactly supported potential
  functions of nodes 1 and 2 respectively. The Lipschitz constant of
  the vector field $(f_1,f_2)$ in Eq.~(\ref{eq:VectorField}) is
  $L=\sqrt{L_1^2+L_2^2}$ where $L_i$ is Lipschitz constant of each
  $f_i$. (The partial derivatives of $\psi_1$ and $\psi_2$ are
  Lipschitz since they are continuously differentiable on a compact
  set.) Define $p^{(0)} = [o_1^{(0)} \quad o_2^{(0)}]$ and $\hat{p}^{(0)} =
  [\hat o_1^{(0)} \quad \hat o_2^{(0)}]$ be two initial conditions to $n$
  steps of Runge-Kutta method $(A,b^T,c)$, using step size $h \le
  h_0$, where $h_0L\rho(|A|) < 1$, and let $p_n$ and $\hat{p}_n$ be the
  corresponding output values, then by \textbf{Lemma 319A} of
  \cite{Butcher} we have:
\begin{equation}
\label{RKbound}
|| p_n - \hat{p}_n || \le (1+hL^*)^n || p_0 - \hat{p}_0 ||
\end{equation}
where $L^* = L|b^T|(I-h_0L|A|)^{-1} \mathbf{1} $ . Letting $h =
\frac{T}{n}$ we have
\begin{equation}
  \dfrac{|| p_n - \hat{p}_n ||}{|| p_0 - \hat{p}_0 || } \le e^{TL^*}
\end{equation}
Since, as $n\rightarrow\infty$, $p_n \rightarrow p^T$ and $\hat p_n
\rightarrow \hat{p}^T$, this last inequality gives us the result.
\end{proof}
\end{prop}

% By mean value theorem, for any differentiable function $g$, we know $|g(x) - g(y)| = g'(z)(x-y)$ for some $z$ between $x$ and $y$. 
% So, if we find supremum of $|g'|$ we are set. In our case $g:= \psi'$, 
% so we have to find supremum of $\psi''$ on $[0,1]$. 

% Then, the differential equation system 
% of the model is given by Eq.~(\ref{eq:VectorField}), therefore 
\begin{exmp} We begin by defining $\gamma = o_1 - o_2$ and the Gaussian potential $\psi_{\mu,\sigma}(\gamma) = \frac{1}{\sigma
    \sqrt{2\pi}}e^{-\frac{(\gamma - \mu)^2}{2\sigma^2}}$. Then
 \begin{equation}\psi_{\mu,\sigma}''(\gamma) = \frac{1}{\sigma^3 \sqrt{2\pi}} e^{-\frac{(\gamma - \mu)^2}{2\sigma^2}} [\frac{(\gamma-\mu)^2}{\sigma^2} - 1]\end{equation}
 We now choose $\psi_1 = \psi_{0.7,1}$ and $\psi_2 = \psi_{0.3,2}$. We
 compute $L_1 = \sup_{|\gamma|\leq 1}|\psi''_1(\gamma)| =
 |\frac{-.3992}{ \sqrt{2\pi}}|$ at $\gamma = 0$ and $L_2 =
 \sup_{|\gamma|\leq 1}|\psi''_2(\gamma)|=|\frac{-.8254}{2^3
   \sqrt{2\pi}}|$ at $\gamma = 1$ and consequently $L = 0.164529$ and
 hence, $L^* = 0.16$.  As an example let $O^{(0)} = [0.2 \quad .1]$ be
 an initial state and two perturbations of it, $O_t^{(0)} = O^{(0)} +
 10^{-8} \bm{v}$ and $O_p^{(0)} = O^{(0)} + 10^{-8} \bm{w}$, where
 $O_t^{(0)}$ is perturbation in direction of tangent line to vector
 field at the point $[0.2 \quad 0.1]$ and $O_p^{(0)}$ is perturbation
 in the direction perpendicular to the tangent line. $\bm{v}$ and
 $\bm{w}$ are unit vectors. Applying RK41 with step size of $h = 0.1$
 and final time $T = 30$ we take $n = T/h = 300$ steps to get to the
 points $O^{(300)}$, $O_t^{(300)}$ and $O_p^{(300)}$.  Let $\bm{v}_1 =
 O_t^{(0)} - O^{(0)}$, $\bm{\hat{v}}_1 = O_t^{(300)} - O^{(300)}$,
 $\bm{w}_1 = O_p^{(0)} - O^{(0)}$ and $\bm{\hat{w}}_1 = O_p^{(300)} -
 O^{(300)}$, then we have \[\cfrac{||\bm{v}_1||}{||\bm{\hat{v}}_1||} =
 0.0363, \quad \cfrac{||\bm{w}_1||}{||\bm{\hat{w}}_1||} = 1.145 \]
 while by Eq.~(\ref{RKbound}) we have
 \[\dfrac{|| O_n - \hat O_n ||}{ || O_0 - \hat O_0 ||} \le (1 +
 0.0165)^{300} = 135.58,\] showing that in this case, the bounds,
 though correct, are very pessimistic.
\end{exmp} 
\begin{remark}
  Note that in the previous example, even though the potential
  functions are no compactly supported in the unit square, they still
  have the property that all derivatives are bounded, which is all the
  we actually used in the previous proposition.
\end{remark}
%%%%%%%%%%%%%%
%%%%%%%%%%%%%% Dynamics of 3 Agents
%%%%%%%%%%%%%%
\subsection{Dynamics of three agents}
\subsubsection{Regions of Fixed Points for Tent Potentials}
In this section we let $G = (V,E)$ be a a fully connected network with
$|E| = 3$ and study the continuous dynamical system generated by those
three interacting agents. All the interaction potentials are assumed
to be identical. The total energy
of such a gradient system is given by:
\begin{equation}\label{eq:totalEnergyDef}
\Psi_{123} = \psi_{12}(|o_1 - o_2|) + \psi_{23}(|o_2 - o_3|) + \psi_{31}(|o_3 - o_1|)
\end{equation}
where $\psi_{ij}$ is the potential assigned to the edge
between nodes $i$ and $j$.  Each interaction gives us
\begin{equation}\label{eq:equalNudgeCont}
\text{$i$-$j$ interaction contribution to gradient vector field }\left\{
	\begin{array}{lr}
    \dot{o}_i(i,j) &= - \psi_{ij}'(|d_{ij}^{(t)}|) \: \frac{d_{ij}^{(t)}}{|d_{ij}^{(t)}|} \\
    \dot{o}_j(i,j) &= - \psi_{ij}'(|d_{ij}^{(t)}|) \: \frac{-d_{ij}^{(t)}}{|d_{ij}^{(t)}|} \\
    \end{array}
\right.
\end{equation}
where $d_{ij} = o_i-o_j$.
% By Eq.~(\ref{eq:equalNudge}), the same step sizes are 
% taken in each direction $o_i$ and $o_j$. Clearly, in the 
% continuous case it means the speed of change is equal, 
% as can be seen easily from Eq.~(\ref{eq:equalNudgeCont}). 
Along the diagonal lines in the coordinate subspace corresponding to
individuals $i$ and $j$, the difference
$o_i - o_j$ is constant and since both agents' behavior is enforced by
the same potential, $\psi_{ij}$, the direction of movement is either
$(1,\:-1)$ or $(-1,\:1)$. See Figure~(\ref{fig:TwoagentSpace}).

\begin{figure}[httb!]
\centering
\includegraphics[width=.6\linewidth]{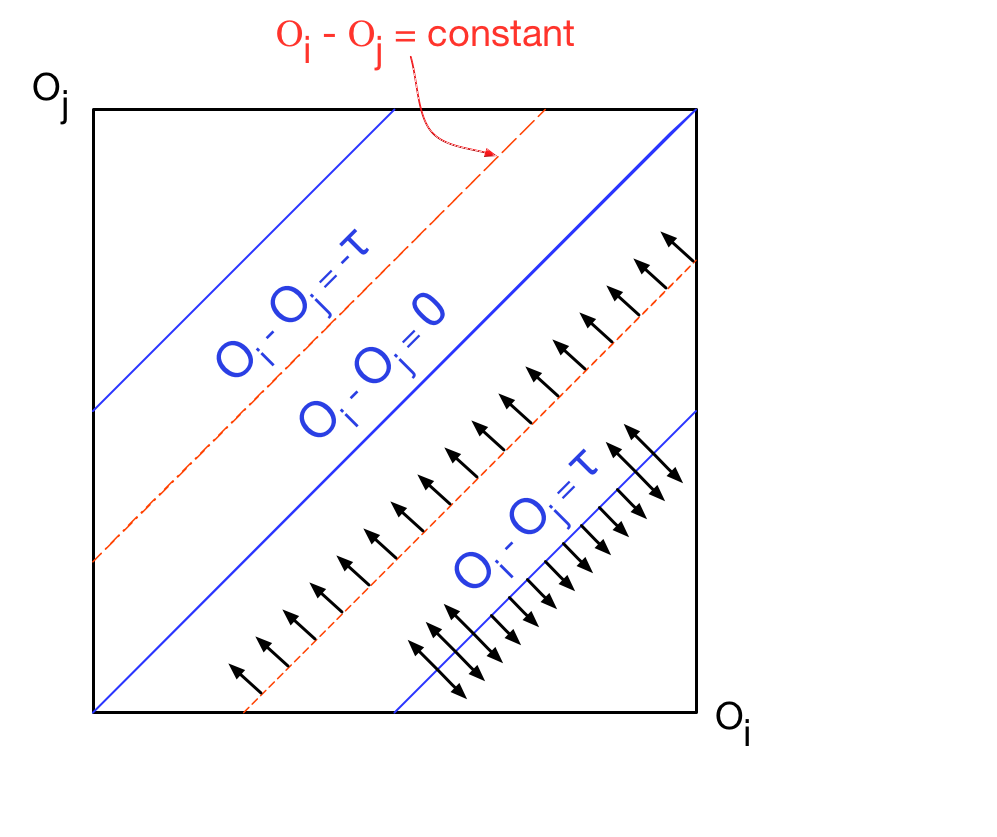}
\caption{Coordinate subspace corresponding to
individuals $i$ and $j$.}
\label{fig:TwoagentSpace}
\end{figure}

By Eq.~(\ref{eq:equalNudgeCont}) we have:
\begin{equation*}
\nabla \psi_{12}(|o_2 - o_1|) = (a,-a,0) , \hspace{.1in} \nabla \psi_{13}(|o_3 - o_1|) = (-b, 0, b) , \hspace{.1in} \nabla \psi_{23}(|o_3 - o_2|) = (0 , c , -c ) 
\end{equation*}
where $a$, $b$ and $c$ can be either positive or 
negative and are given by the right side of Eq.~(\ref{eq:equalNudgeCont}). 
Combining these, we get that the gradient vector field generated by our potential energy is given by:
\begin{equation} \label{eq:totalGrad}
\begin{array}{lllllll}
\dot{\vv{\mathbf{o}}} &=& \nabla \Psi_{123} \\
             &=& \nabla \psi_{12}(|o_1 - o_2|) + \nabla  \psi_{13}(|o_1 - o_3|) + \nabla   \psi_{32}(|o_3 - o_2|)\\
             &=& ( a ,\: -a,\: 0) + ( 0 ,\: b,\: -b) + ( -c ,\: 0,\: c) \\
             &=& ( a - c, \: b - a, \: c - b).
\end{array} 
\end{equation}
Equation~(\ref{eq:totalGrad}) shows we will get fixed points if $a = b = c$. 
\bigskip

Denoting the opinion state space by $\mathcal{C}_3 = \Bbb{O}^3 = [0,1]^3$,
we define \[\mathcal{D}_3 = \{(o_1, o_2, o_3)| o_i = c \in [0,1], 1 \le i
\le 3 \}\] to be the diagonal of $\mathcal{C}_3$ corresponding to
consensus states, denote the set of extreme points where the system is
polarized by \[\mathcal{E}_3 = \{ (1,0,0), (0,1,0), (0,0,1), (1,1,0),
(1,0,1), (0,1,1) \}.\]

The next theorem shows that when the common potential is a tent
function, there are regions of neutral fixed points in the system that
are neither consensus points nor polarization points.

\begin{theorem}\label{fixedPointRegion}
  Let $G_3 = (V, E)$ be a fully connected network of three nodes. Let
  all interaction potentials be the same tent potential function, with
  its tip at $\tau = 0.5$ and maximum height $h$. Then there are
  subsets of $\mathcal{C}_3\setminus (\mathcal{D}_3 \cup \mathcal{E}_3)$
  with nonempty interior, which are made up of neutral fixed points.
\end{theorem}

\begin{proof} Without loss of generality, we will consider the case
  $o_1 < o_2 < o_3$, other cases follows similarly. Let the point $p =
  (o_1, o_2, o_3)\in\mathcal{C}_3$ such that $d_{12} = o_2 - o_1 =
  \tau - \epsilon = o_3 - o_2 = d_{23}$ and $d_{13} = o_3 - o_1 = 2
  (\tau - \epsilon) > \tau$. Since the potential function has the same
  slope in magnitude on both sides of $\tau = 0.5$, by
  ~(\ref{eq:totalGrad}) we have:

\begin{equation} 
\begin{array}{lllllll}
\dot{\vv{\mathbf{o}}} &=& \nabla \Psi_{123} \\
             &=& ( \alpha h ,\: -\alpha h,\: 0) + ( 0 ,\: \alpha h,\: -\alpha h) + ( -\alpha h ,\: 0,\: \alpha h) \\
             &=& ( 0, \: 0, \: 0)
\end{array} 
\end{equation}
Moreover, the point $P$ is not an isolated 
fixed point. Notice that the point $p$ can 
be perturbed to $\hat{p} = p + \zeta (1, 1, 1)$ 
and still be a fixed point as long as 
$ \tau < \hat{d}_{13} < 1 $ and $ 0 < \hat{d}_{12}, \hat{d}_{23} < \tau$, 
where $\hat{d}_{ij} = | \hat{o}_i  - \hat{o}_j| $.
\end{proof}

\begin{exmp} Let $G_3$ be a fully connected 
network with potential tent functions, 
$\tau = 0.5$ and $h = 1$, for each edge. 
Let $(o_1,\: o_2,\: o_3) = (0.2, \: 0.5, \: 0.8)$. 
Then each agent is pushed and pulled by 
the same force, $ a = b = c = 1$ above, by 
the other two agents, and therefore they 
do not move, hence, we have a fixed 
point and indeed a region of fixed points.
\end{exmp}
In the Thm~(\ref{fixedPointRegion}) we 
considered only the case $o_1 < o_2 < o_3$. 
There are five other possible cases, and all 
these six regions are mutually exclusive. 
The regions are given in Sec. (\ref{sec:Appendix}) 
and it is shown that their interior are mutually disjoint and
 afterward we show that there must be 
 fixed points inside these regions.

\bigskip

These fixed points are not specific to three 
agent systems. In fact, we can immediately extend the results to the case in 
which there are $N = 3k$ agents. 
\begin{corr}\label{TheCorollary} Let $G = (V,E)$ be a fully connected graph with $|V|= N =
  3k$ for some $k \in \mathbb{N}$ and let all the agents have the same
  tent potential function $\psi$ with peak at $0.5$.  Then there the
  are subsets of
  $\mathcal{C}_N\setminus(\mathcal{E}_N\cup\mathcal{D}_N)$ with
  non-empty interior which are comprised of fixed points.

\begin{proof}
Let $p = (o_1, \: o_2, \: o_3)$ be a 
neutrally stable fixed point in the $R_1$ derived 
above. Moreover, let $S_i$ be a set of $k$ agents in the
$\delta_i$ neighborhood of $o_i$, $ i \in \{1, \: 2, \: 3 \}$. 
Furthermore, let  
\begin{align*}
 \max\{S_2\} - \min\{S_1\} < 0.5 - \epsilon \\
 \max\{S_3\} - \min\{S_2\} < 0.5 - \epsilon \\
 \min\{S_3\}  - \max\{S_1\} > 0.5 + \epsilon
\end{align*}
% and where $|o_i - o_j| \not\in (0.5 - \epsilon, 0.5 + \epsilon)$. 
 Then since the potential function has constant 
 and the same slope size on the 
 whole region except $ (0.5 - \epsilon, 0.5 + \epsilon)$, 
 the force applied to nodes in $S_i$ by nodes in $S_j$ 
 and nodes in $S_k$- $i \neq j \neq k $- would cancel out 
 and the only force working on nodes in the $S_i$ 
 is applied by nodes in the $S_i$, and therefore, 
 they all would come to consensus at a point 
 inside the $\delta_i$ neighborhood of $o_i$, 
 denote it by ${o_i^*}$ and therefore, 
 $p^* = (o_1^*, \: o_2^*, \: o_3^*)$ would be a fixed point inside the $R_1$.
\end{proof}
\end{corr}

\bigskip

\begin{remark}
  In the case of this potential $\Psi_{123}$ we know that when we are
  not in at a fixed point, the rate of energy decrease is bounded
  below, so we are guaranteed to converge to a fixed point in time
  less than or equal to $\frac{\sup_{\mathcal{C}_3} \Psi_{123}}{\alpha h}$.
\end{remark}

Next, we look at bell-shaped potentials and the fixed 
points they generate using a simple example potential.

\subsubsection{Bell shaped potential}
Take a symmetric potential function about $x = 0.5$ 
with zero slope at $x = 0$ and $x = 1$. We can obtain 
this by taking a smooth function between zero and 
$0.25$ and then rotate, reflect and shifting it to make 
up the function, like Fig.~(\ref{fig:rotatedPotential}), 
which corresponds to the function by Eq. (\ref{eq:RotateFunct}).

\begin{equation} \label{eq:RotateFunct}
\psi(x) = \left\{
  \begin{array}{lr}
    \vspace{.1in}
    x^2   &  \hspace{.1in} 0 \leq x \leq 0.25\\
    \vspace{.1in}
    - ( x - \frac{1}{2} )^2 + \frac{1}{8}  & \hspace{.1in} 0.25 < x < 0.75 \\
      ( x - 1)^2  &  0.75 \leq x \leq 1\\
  \end{array}
\right.
\end{equation}
\begin{figure}[httb!]
  \centering
\includegraphics[width=.4\linewidth]{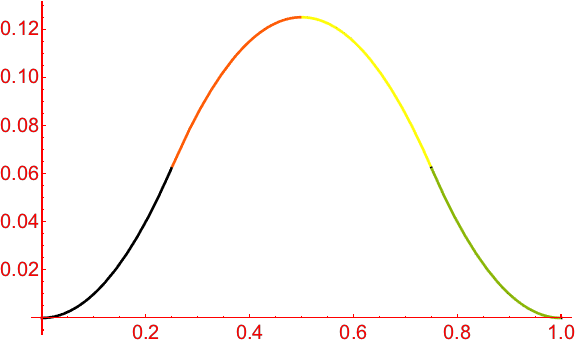}
\caption{Potential function by rotations, reflections and shifts of $x^2$}
\label{fig:rotatedPotential}
\end{figure}
Next we show a symmetric bell-shaped potential function has lines of neutrally stable fixed points.
%%%%%%%%%%%
%%%%%%%%%%% Bell Shaped Theorem
%%%%%%%%%%%
\begin{theorem}\label{Bell-Shaped-Fixed-Points}
Let $G = (V, E)$ be a fully connected network of 3 
agents and let each agent to have the same 
bell-shaped potential $\psi$. Then there exist 
six line segments of fixed points. Moreover, 
these regions can be turned into unstable 
fixed point regions by perturbation of the potential function.

\begin{proof}
Let $p = (o_1^*, \: o_2^*, \: o_3^*)$ be a given fixed point in 
$(\mathcal{D} \cup \mathcal{E})'$ where $o_1^* < o_2^* < o_3^*$. 
Moreover, let $ f_{ij} =  \frac{\alpha}{2} \psi'_{ij}$ be the force imposed 
by agent $i$ on agent $j$ where $\psi'_{ij} = \psi'(d_{ij}) $. 
Since $p$ is a fixed point the forces $f_{ij}$ in Fig.~(\ref{fig:BellForces}) 
have the same magnitude.  Since $f_{21} = f_{12}$ we must have 
$f_{12} = f_{32}$ in order for agent two to not move. Therefore, 
$d_{12} = d_{23} = d \in (0.25, \: 0.5)$ and $d_{13} = 2d \in (0.5, \: 0.75)$. 
Note that $d_{12}$ and $d_{23}$  cannot be in 
$(0, \: 0.25)$ because otherwise, 
$d_{13} \in (0, \: 0.5)$  and there would 
not exist a repelling force. Since 
$2d \in (0.5, 0.75)$ we have to have $d \in (0.25, \: 0.375)$. 

\begin{figure}[httb!]
  \centering
\includegraphics[width=.6\linewidth]{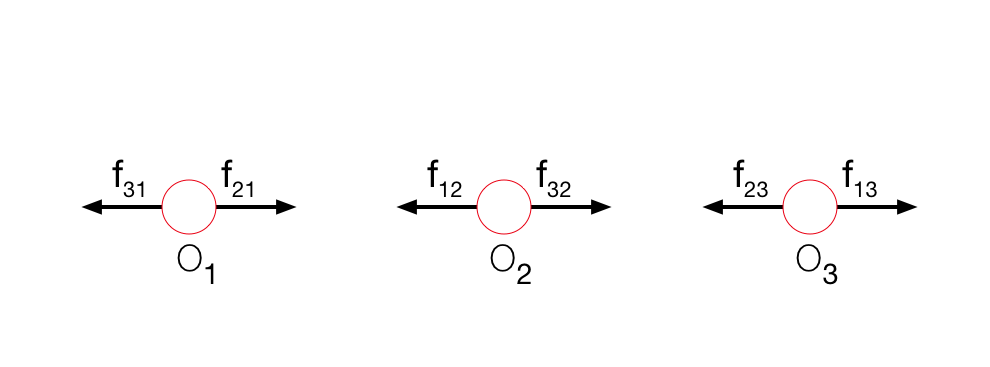}
\caption{}
\label{fig:BellForces}
\end{figure}

Starting from agent 1, we want to have $F_1 = f_{21} + f_{31} = 0$. 
Since the potential function is symmetric about $x = 0.5$ 
there exist the points $d \in (0.25, \:0.5)$ 
such that $\psi'(d) = -\psi'(2d)$ and 
consequently $F_1 = 0$. The same is true for $F_2$ and $F_3$. 

Neutrality of $p$ is shown below. 
In our system we have $\dot{\vv{\mathbf{o}}} = \nabla \psi_{123}$. 
Denote different components as follows:
\begin{equation} \label{eq:VFfunction1}
  \left\{ \begin{array}{lr}
    \vspace{.1in}
    \dot{o}_1  &= f(o_1,o_2,o_3)\\
    \vspace{.1in}
    \dot{o}_2  &= g(o_1,o_2,o_3)\\
    \dot{o}_3  &= h(o_1,o_2,o_3)\\
  \end{array}
  \right .
\end{equation}
Let $p = (o_1^* , \: o_2^* , \: o_3^*)$ be a fixed 
point of the three-agent interactions. To show 
these are unstable fixed points we would use 
linearization of the potential about the fixed 
point. Let $u = o_1 - o_1^*$, $v = o_2 - o_2^*$ 
and $w = o_3 - o_3^*$ be small perturbation 
of the fixed point. To see whether the 
perturbation grows or decays we have 
to derive differential equations for $u$, $v$ and $w$:
\begin{equation} 
\begin{array}{lllllll}
\dot{u} = \dot{o_1} &=& f(o_1^* + u, o_2^* + v, o_3^* + w) \\
             &=& f(o_1^*, o_2^*, o_3^*) + f_{o_1} u + f_{o_3} v + f_{o_3} w + \frac{1}{2!} \begin{bmatrix}
         u &  v & w 
        \end{bmatrix} H(o_1, o_2 ,o_3 )\begin{bmatrix}
         u \\
         v \\
         w\\
        \end{bmatrix} + \cdots\\
        &=&  f_{o_1} u + f_{o_2} v + f_{o_3} w  + O( u^2, v^2, w^2, uv, wv, uw) \\
\end{array} 
\end{equation}
Where $H(o_1, o_2 ,o_3) = \begin{bmatrix}
    f_{o_1o_1} & f_{o_1o_2} & f_{o_1o_3} \\
    f_{o_2o_1} & f_{o_2o_2} & f_{o_2o_3}  \\
    f_{o_3o_1} & f_{o_3o_2} & f_{o_3o_3}
\end{bmatrix}$, called the Hessian matrix and $f_{o_i}$ is derivative of $f$ with 
\vspace{.1in}
respect to $o_i$ evaluated at $( o_1^* , o_2^*, o_3^*)$. The same can be obtained for $v$ and $w$. And perturbation evolves according to
\begin{equation}\label{eq:LinearizedEq}
\begin{array}{lllllll}
\begin{bmatrix}
         \dot{u} \\
         \dot{v} \\
         \dot{w}\\
        \end{bmatrix} &=& \begin{bmatrix}
    f_{o_1} & f_{o_2} & f_{o_3} \\
    g_{o_1} & g_{o_2} & g_{o_3} \\
    h_{o_1} & h_{o_2} & h_{o_3}
\end{bmatrix}_{(o_1^*, o_2^*, o_3^*)} \begin{bmatrix}
         u \\
         v \\
         w\\
        \end{bmatrix} +  O( u^2, v^2, w^2, uv, wv, uw)\\
        &=& J(o_1^*, o_2^*, o_3^*)\begin{bmatrix}
         u \\
         v \\
         w\\
        \end{bmatrix} +  O( u^2, v^2, w^2, uv, wv, uw)\\
\end{array}
\end{equation}
which by dropping the quadratic terms we get the linearized system. 
Dropping the learning rate $\alpha$, we have:
\begin{equation} 
\begin{array}{lllllll}
\dot{o} = \nabla \Psi_{123} &=& \nabla \psi_{12} + \nabla \psi_{32} + \nabla \psi_{31}  \\                      &=& (\psi'_{12}(|o_2 -  o_1|) \frac{o_2 - o_1}{|o_2 - o_1|}, \: - \psi'_{12}(|o_2 -  o_1|) \frac{o_2 - o_1}{|o_2 - o_1|}, \: 0) + \\
             & & (0, \:  \psi'_{32}(|o_3 -  o_2|) \frac{o_3 - o_2}{|o_3 - o_2|}, \: -\psi'_{32}(|o_3 -  o_2|) \frac{o_3 - o_2}{|o_3 - o_2|}) + \\
             & & ( -\psi'_{31}(|o_3 - o_1|) \frac{o_3 - o_1}{|o_3 - o_1|} , \: 0, \: \psi'_{31}(|o_3 - o_1|) \frac{o_3 - o_1}{|o_3 - o_1|}) \\
             & &\\
             \vspace{.1in}
             
             &=& 
             \begin{bmatrix}
         \psi'_{12}(|o_2 -  o_1|) \frac{o_2 - o_1}{|o_2 - o_1|}  -\psi'_{31}(|o_3 - o_1|)    \frac{o_3 - o_1}{|o_3 - o_1|} \\
         \psi'_{32}(|o_3 -  o_2|) \frac{o_3 - o_2}{|o_3 - o_2|} - \psi'_{12}(|o_2 -  o_1|) \frac{o_2 - o_1}{|o_2 - o_1|} \\
         \psi'_{31}(|o_3 - o_1|)  \frac{o_3 - o_1}{|o_3 - o_1|} -\psi'_{32}(|o_3 -  o_2|) \frac{o_3 - o_2}{|o_3 - o_2|}\\
        \end{bmatrix}^T\\
        &=& \begin{bmatrix}
         f(o_1, o_2, o_3) \\
         g(o_1, o_2, o_3) \\
         h(o_1, o_2, o_3) \\
        \end{bmatrix}^T
        
\end{array} 
\end{equation}
Therefore, the Jacobian matrix of~(\ref{eq:LinearizedEq}) is given by
\[\small{
\mathbf{J} = \begin{bmatrix}
    \psi''_{13}(|o_3-o_1|) - \psi''_{12}(|o_2-o_1|) & \psi''_{12}(|o_2-o_1|) & -\psi''_{13}(|o_3 -o_1|) \\
    \psi''_{12}(|o_2-o_1|) & -\psi''_{23}(|o_3 - o_2|) - \psi''_{12}(|o_2-o_1|) & \psi''_{23}(|o_3 - o_2|) \\
    -\psi''_{13}(|o_3 - o_1|) & \psi''_{23}(|o_3 - o_2|) & \psi''_{13}(|o_3 - o_1|) - \psi''_{23}(|o_3 - o_2|)
\end{bmatrix}
}
\]
Since $d_{12} = d_{32}$ the Jacobian has the following structure:
\[
\mathbf{J} = \begin{bmatrix}
  b - a &  a    & -b \\
  a     & -2a  &  a \\
 -b &      a    & b-a
\end{bmatrix}
\]
The Jacobian matrix is a real symmetric matrix, 
and therefore, the eigenvalues are real, and 
therefore, centers or spirals do not exist. 
Since there is a line of fixed points, we 
cannot have three dimensions  all going 
to the fixed points. That can also be seen 
from the structure of the Jacobian matrix. 
Due to the  structure of the matrix, determinant 
of $J$ is zero, therefore, one of the eigenvalues 
is zero, which indicates fixed points are not isolated. 
Moreover, in a real symmetric matrix, number 
of negative eigenvalues is the same as number of negative pivots.
 The Jacobian matrix $J$ given above always have two negative eigenvalues which can
 be easily seen by looking at pivots signs.
 
 However, we can perturb 
 the potential function given by Eq.~(\ref{eq:RotateFunct}) 
 using non-analytical functions at the given $d_{12}, d_{13}, d_{23}$, so that the slope is kept 
 the same, but the second derivative is changed so that the 
 Jacobian matrix has one positive eigenvalue, one negative
and a zero eigenvalue. (Two positive eigenvalues is impossible!)
 
We proved the theorem for $o_1 < o_2 < o_3$, the other five cases 
 follows similarly.
\end{proof}
\end{theorem}
%%%%%%%%%%%%
%%%%%%%%%%%% Bell Shaped Example
%%%%%%%%%%%%
\begin{exmp} 
Let the potential function of the three agents be given by Eq.~(\ref{eq:RotateFunct}). 
The fixed points in $(\mathcal{D} \cup \mathcal{E})'$ are of the form $p = (o_1^*, \: o_2^*, \: o_3^*) = ( o - \frac{1}{3}, \: o, \: o + \frac{1}{3})$.
The line segment generated by $o$ where $o - \frac{1}{3} \geq 0 $ and  $o + \frac{1}{3} \leq 1$ is the line segment of neutrally stable fixed points since:
\[F_1 = f_{21} + f_{31} = \frac{\alpha}{2} \left( \frac{1}{3} -   \frac{1}{3} \right) = 0 = F_2 = F_3 \]
For example, if $p = (1/6, 1/2, 5/6)$, then 
$J(o_1^*, o_3^*, o_3^*)= \frac{1}{3}\begin{bmatrix}
  -2 & 1 &  1 \\
   1 & -2 & 1 \\
   1 & 1 &  -2 
\end{bmatrix}.$
Eigenvalues of the Jacobian are $\sigma(\mathbf{J}) = \{-3, 0 \}$.
\end{exmp}

\begin{remark}
Note that there are also fixed points of the forms 
$p = (a, \: a, \: a + 0.5)$, \linebreak $q = (a, a+\: 0.5, \: a)$ 
and $r = (a + 0.5, \: a, \: a)$ due to having derivative zero 
at end points and in the middle of the potential.
\end{remark}

\begin{remark}
  The bell shaped potential above is just one potential in a family of
  potentials determined by strictly increasing smooth functions
  $f:[0,0.25] \rightarrow \mathbb{R}^+ \cup \{0\}$ such that \linebreak $f(0) = f'(0) = 0$. The fact
  that the above potential generates a stable line of fixed points
  means that small enough perturbations of the bell shaped potential
  retain the non-consensus, non-polarized fixed point. Furthermore,
  using these bell-shaped potentials, we can obtain fixed points for
  systems with $N = 3k$ agents, as was done with the tent potential in
  Corollary~\ref{TheCorollary}.
\end{remark}

%%%%%%%%%
%%%%%%%%% Discrete Game properties
%%%%%%%%%
\subsection{Discrete game properties}
\begin{prop}\label{LiapunovDance}
  Let $G_3$ be a fully connected network of three nodes and let $\psi$
  be a tent potential function centered at $\tau = 0.5$, with maximum
  $h$ and a learning rate $\alpha$. Let $p = (\hat{o}_1, \:
  \hat{o}_2, \: \hat{o}_3)$ be a given point in the unit cube so that:
  \begin{enumerate}
  \item $0 < d^{(0)}_{12}, d^{(0)}_{23} < 0.5$
  \item $d^{(0)}_{13} > 0.5$
  \item $\alpha h < \min \{ 0.5 - d^{(0)}_{23}, \frac{d^{(0)}_{13} -
      0.5}{2} \}$ and
  \item the deterministic, asynchronous game order is
    $(1,2)\rightarrow(2,3)\rightarrow(3,1)\rightarrow(1,2)\rightarrow\cdots$
  \end{enumerate}
  Then $p$ and all the points in an open neighborhood of p are
  neutrally stable in the sense that $p^{(0)} = p^{(3)} = p^{(6)} =
  \cdots = p^{(3k)} = \cdots$. I.e. all points in that neighborhood
  define periodic orbits of with period 3.
\end{prop}

\begin{proof}
Since the step size $\alpha h < 0.5 - d^{(0)}_{23}$, 
after node 1 playing with node 2, we have $d^{(1)}_{23} < 0.5$, 
and therefore, nodes 2 and 3 attract each other, and node 2 
is back on its starting point. And since $\alpha h < \frac{d^{(0)}_{13} - 0.5}{2}$, 
we have $d^{(2)}_{13} > 0.5$, and therefore nodes 1 and 3 repel 
and they both are back on their initial positions. The loop would 
goes forever. Not only $p$ generates this periodic state, but also 
a neighbor of $p$, $(\hat{o}_2 + \eta_1, \: \hat{o}_2 + \eta_2 , \: \hat{o}_3 + \eta_3 )$ , 
in which the conditions on the $d^{(0)}_{ij}$ and $\alpha h$ 
are satisfied would be a periodic path. 
\end{proof}

\begin{remark}
  The discrete, asynchronous, deterministic game will approximate the
  continuous synchronous game as long as we stay away from the
  boundaries of the 6 neutral regions: how close we can get depends on
  how small the learning rate is.  If the step sizes in the game we
  are playing cause us to step across the boundary of the neutral
  region, then we converge to either consensus or polarization,
  depending on which boundary we step across. A game in which the
  play order is not deterministic, but is rather random, would
  generate a random walk that would eventually escape the neutral
  region.
\end{remark}

%%%%%%%%%%%%%%%%%%%%%%%%%%%%%%%%%%%%%%%%%%%%%%%%%%%%%%%%

%%
%%
%%
%\section{Theorems}
% {\color{red}{The following are not really theorems, but observations.}}

% \begin{theorem}
% The final state of the system does not solely depend on initial settings and it heavily depends on order of negotiations as well.
% \end{theorem}

% %\begin{exmp}
% %In the following examples the population is 15, initial opinions are the same, learning rate is fixed, $\alpha = 0.1$, and the potential function used is Gaussian, where the network is a full graph. And we can see the final states are different.
% %\end{exmp} 

% %%
% %% Initial Final figures
% %%

% \begin{theorem}
% The convergence time, $t_{conv}$, does not depend solely on initial settings. It also depends on order of negotiations.
% \end{theorem}

% %\begin{exmp}
% %In the following full graph examples the population is 15, initial opinions are the same, learning rate is fixed, $\alpha = 0.1$, and the potential function used is Gaussian. The time of convergence is different, even in the case that final states are the same.
% %\end{exmp} 

% %%
% %% convergence time
% %%

% %%%%%%%%%%%%%%%%%%%%%%%%%%%%%%%%%%%%%%%%%%%%%%%%%%%%%%%%%%%%%%%%%
% %%%%%%%%%%%%%%%%%%%%%%%%%%%%%%%%%%%%%%%%%%%%%%%%%%%%%%%%%%%%%%%%%
% %%%%%%%%%%%%%%%%%%%%%%%%%%%%%%%%%%%%%%%%%%%%%%%%%%%%%%%%%%%%%%%%%
% %%%%%%%%%%%%%%%%%%%%%%%%%%%%%%%%%%%%%%%%%%%%%%%%%%%%%%%%%%%%%%%%%
% %%%%%%%%%%%%%%%%%%%%%%%%%%%%%%%%%%%%%%%%%%%%%%%%%%%%%%%%%%%%%%%%%

\section{Experimental Results}
\label{experiments}
In this section we present experiments\footnote{The codes for this experiments can be found here: https://github.com/HNoorazar/} examining the behavior of
networks whose parameters, such as population or topology or
interaction potential parameters, are fixed except one. Unless
otherwise is stated, the stopping condition is the window-convergence
(introduced below). In these preliminary studies, we are interested
in understanding how changing a parameter changes the average
stabilization time and the population of the polarization/consensus
clusters.

\begin{remark}
In the experiments we use the update rule $o_i^{(t+1)} =
o_i^{(t)} - \frac{\alpha} {2} \psi'(|d_{ij}|) d_{ij}$ rather than
$o_i^{(t+1)} = o_i^{(t)} - \frac{\alpha} {2} \psi'(|d_{ij}|)
\frac{d_{ij}}{|d_{ij}|} $ to prevent over shooting. This might be
considered as an adaptive learning rate.  
\end{remark}

\subsection{Convergence and stability}

One of the primary goals of an opinion dynamics model is to understand
the state of the system in the limit of an arbitrary number of time
steps. Does the system reach a stable steady state, oscillate between
a finite number of deterministic states, or is the outcome stochastic
or unstable?  In the case where multiple final states can be reached,
what is the probability of any given state being reached relative to
others?  To reason about this we must define the concepts of
convergence and stability with respect to the opinion
game. Convergence relates to a slowing of change within the system,
where we would say that the system has converged to a steady state if
the opinion state across the entire system has stopped
changing. \emph{Consensus} is a special case of convergence in which
the system has converged and the opinion state across the individuals
is one in which every individual agrees with the others. This is the
final state of the DeGroot averaging model. Our model admits converged
states in which consensus is not present (e.g., reaching a final state
in which opinions are split between two sub-populations). Stability
relates to the sensitivity of the system to changes when it has
reached a steady state. A stable system will be able to tolerate some
degree of change while converged, while an unstable system may enter
into a non-steady state when a change occurs and ultimately end up in
a new converged final state.

\theoremstyle{definition}
\subsubsection{Definition of Window-convergence}

In order to decide when a system has converged, one could just wait
until nothing changes in a sequence of interactions involving all
pairs, but this is usually not a practical approach. Instead, we will use a
notion of convergence (or pseudo-convergence) in which we continuously
observe changes in some \emph{moving window in time} and stop when
some criterion has been met.

\bigskip

Let $\omega$ be the length of that time interval over which we monitor
opinion changes. Assume there is one topic we are tracking over
time. Let $E_s$ be a matrix of agents opinions from time $t = 1$ to
time $t = s$, where the $k^{th}$ column contains the state vector of
the opinions of the $N$ agents on the topic at time $k$.  In this
approach to a stopping criterion, we consider only the last $\omega$
columns of $E_s$, which we denote by $S_s$': I.e. $S_s =
E(1:N,s-\omega+1:s)$. Define the $N$ dimensional column vector $F_s$:
\[F_s(i) = \max\{\max_{j=1,...,\omega} S_s(i,j) - \min_{j=1,...,\omega}
S_s(i,j) - \eta, 0\}.\]

\bigskip

\noindent We stop when F is the zero vector. 

% \textit{Collapsed-E}
% Define a vector of size $N$ which I refer to as \textit{Collapsed-E}
% whose entries are infinity at $t=0$. After ``$\omega +
% \textit{control-length}$" iterations we generate a vector of size $N$
% by subtracting the minimum entries of each row from the maximum entry
% of the same row of the matrix $E$. If the subtraction is less than a
% given threshold, then the corresponding entry of the vector
% \textit{Collapsed-E} is set to zero and by the time all entries of
% \textit{Collapsed-E} are zero the game stops.

\subsection{Effect of Initial Opinion}
In begin to see how initial opinions effect stabilization time, we
generated initial opinions using normal distributions with a mean of
$0.5$ and a range of variances. In this case, the game was played by
picking pairs at random and calculating changes in opinion for each
node/agent based on the variational model determined by a tent function
with peak at $0.5$.

Define $\mu_t$ to be the mean time to stability, $\sigma_t$ the
variance in the stabilization times, and $\mu_p$ to be the mean
fraction of final states that are polarized (as opposed to converging
to consensus). The following figures plot (a) initial opinion variance
versus the mean and standard deviation of stabilization time, $\mu_t$
and $\sigma_t$, and (b) initial opinion variance versus the mean
polarized fraction, $\mu_p$.

For a given $\sigma$, 100 initial opinions were drawn from the normal
distribution with mean 0.5 and standard deviation $\sigma$,
$N(0.5,\sigma)$. For every initial opinion, the experiment was run 100
times to begin to average out the effect of the random order in which
the game was played.

\begin{figure}[httb!]
  \centering
\begin{subfigure}{.35\textwidth}
  \includegraphics[width=1\linewidth]{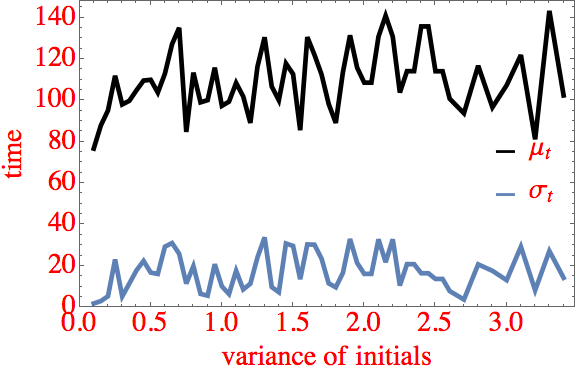}
  \caption{Effect of initial state on stabilization time}
  \label{fig:stepPolFull}
\end{subfigure}%
\hspace{.5in}
\begin{subfigure}{.35\textwidth}
  \includegraphics[width=1\linewidth]{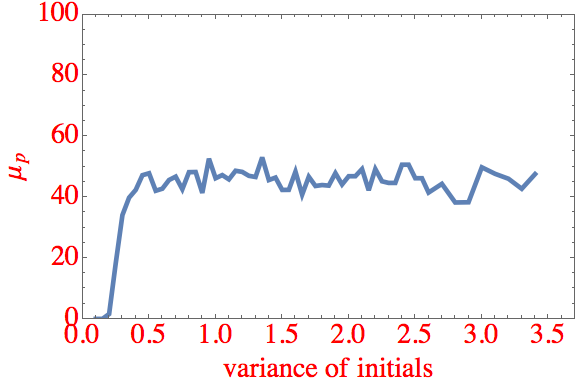}
  \caption{Effect of initial state on polarization count}
  \label{fig:stepTimeFull}
\end{subfigure}%
\caption{Initial state effects}
\label{fig:InitialEffect}
\end{figure}

% If the game evolves stochastically asynchronously, then each node is traversing a random walk, and would finally leave the region.

\subsection{Effect of Tent Tip}
\label{full-experiments-vary-potential}

The experiments in this section are done with a fully connected
interaction graph and a tent potential function where $\tau$, the peak
of the tent, moves between 0 and 1.  Three experiments were run, each
with a different strategy for picking initial conditions:

\begin{enumerate}
\item \emph{Step initials}, in which the opinions were spread out
between zero and one equally, with the same step size between them,
\item \emph{Uniform initials}, in which initial opinions were sampled from
  the uniform distribution on $[0,1]$, and
\item \emph{Normal initials}, in which initial opinions were sampled
  from a normal distribution. The desired
  number of samples were drawn from a normal distribution with mean 0 and
  standard deviation 1, and then shifted and scaled so that the samples
  all lie within $[0,1]$.
\end{enumerate}
In each of these experiments there were 20 nodes/agents were involved and 
for each initial condition, the game is played to convergence 1000 times.
\subsubsection{Step Initials} 
Table~\ref{table:stepInitPot} and Figure~\ref{fig:StepPeak}\subref{fig:stepPolFull} show how polarization varies as the tent peak
$\tau$ moves from zero to 1. Because there is no randomness in the initial
conditions, each experiment is run 1000 times for each position of the
tent peak. This allows the effects of the randomness in the order of game
negotiations to be averaged out.

Observe that we get a probability of polarization of $0.5$ when
$\tau \approx 0.63$. Because $\tau=0$ implies always-polarization and $\tau=1$
implies always-consensus, we might be tempted to expect that at $\tau=0.5$ we
would have a probability of polarization of $0.5$. But this is a
result of identifying the domain of the potential with the opinion
space, even though this is not correct. The potential is a function of
opinion differences, not the opinions themselves. A deeper look at
this reveals that the important factors determining equilibrium states
are volumes of attractors in the opinion state spaces which are
controlled by the location of the tip of the tent. This is part of our ongoing
work to be included in a subsequent paper to follow this one.

\begin{figure}[httb!]
  \centering
\begin{subfigure}{.35\textwidth}
  \includegraphics[width=1\linewidth]{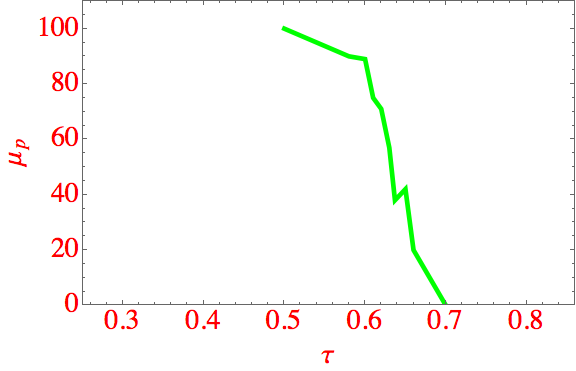}
  \caption{Polarization mean}
  \label{fig:StepPolPeak}
\end{subfigure}%
\hspace{.6in}
\begin{subfigure}{.35\textwidth}
  \includegraphics[width=1\linewidth]{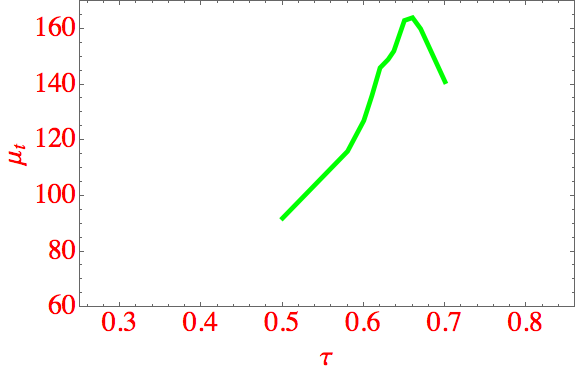}
  \caption{Mean of stabilization time}
  \label{fig:StepTimePeak}
\end{subfigure}%
\caption{Effect of moving $\tau$ on step initials.}
\label{fig:StepPeak}
\end{figure}
%%
%% Table of effect of peak on step initials
%%
\begin{table}[!htb]
 \caption{Step initial opinions with different tent potentials with $\tau \in [0.5, 0.7]$.}
\begin{center}
    \begin{tabular}{| l | l | l | l | l | l| l| l| l | l | l | p{1cm} |}
    \hline
    $\tau$  & 0.5 & 0.58 & .6  &  0.61 & 0.62 & 0.63 & 0.637 & 0.65 & 0.66 & 0.67 & .7 \\ \hline
    Polarization count & 100 & 90 & 89 & 75 & 71 & 57 & 38 & 42 & 20 & 15 & 0 \\ \hline
    Stabilization time  & 92 & 116 & 127 & 136 & 146 & 149 & 152 & 163 & 164 & 160 & 141  \\ \hline
    \end{tabular}
\end{center}
 \label{table:stepInitPot}
\end{table}

\subsubsection{Uniform initials}
In this experiment we sampled the initial opinions from a uniform
distribution over $[0,1]$ 100 times, and for each sample the
simulation is run 1000 times.  As a result, both the order of
negotiations and effect of initial sampling are taken into effect in
the resulting statistics. The same is done for initial opinions
sampled from normal distribution. Thus, for each value of $\tau$, 100,000
experiments are run.  In Table~\ref{table:unifPotTable} we report the
normalized fraction of those experiments that resulted in
polarization by dividing the total number of polarizations by 1000.
%%%
\begin{figure}[httb!]
  \centering
\begin{subfigure}{.35\textwidth}
  \includegraphics[width=1\linewidth]{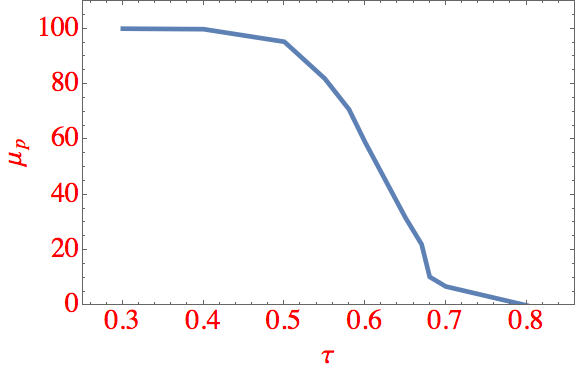}
  \caption{Polarization mean}
  \label{fig:UniformPolPeak}
\end{subfigure}%
\hspace{.6in}
\begin{subfigure}{.35\textwidth}
  \includegraphics[width=1\linewidth]{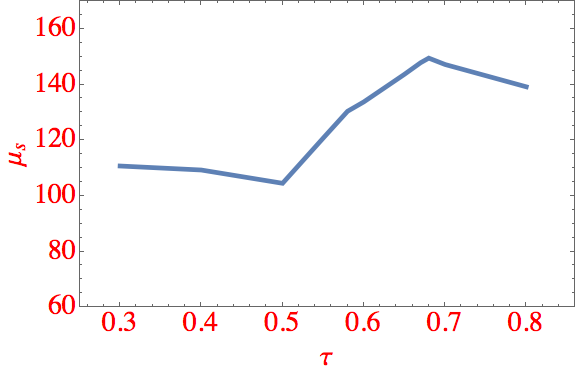}
  \caption{Mean of stabilization time}
  \label{fig:UniformTimePeak}
\end{subfigure}%
\caption{Effect of $\tau$ on uniform initials}
\label{fig:UniformPeak}
\end{figure}
%%%
\begin{table}[!htb]
 \caption{Uniform initial opinions with different tent potentials with $\tau \in [0.3,0.8]$.}
\begin{center}
    \begin{tabular}{| l | l | l | l | l | l| l| l| l | l | l | p{1cm} |}
    \hline
    $\tau$  & .3 & .4  & .5 &  .55 & .58 & .6 & .65 & .67 & .68 & .7 & .8  \\ \hline
   $\mu_p$     & 100 & 99.81 & 95.27 & 81.98 & 70.87 & 59 & 31.69 & 22.08 & 10.3 & 6.82 & 0  \\ \hline
   $\mu_s$  & 110.73 & 109.29 & 104.52 & 120.75 & 130.39 & 133.76 & 143.66 & 147.84 & 149.48 & 147.20 & 139.18  \\ \hline
    \end{tabular}
\end{center}
 \label{table:unifPotTable}
\end{table}

\subsubsection{Normal initials}

The initial 20 opinions here are generated by sampling the standard
normal distribution ($\mu =0$, $\sigma = 1$), and then 
shifted and scaled so that they fit in the interval $[0,1]$.

\begin{figure}[httb!]
  \centering
\begin{subfigure}{.35\textwidth}
  \includegraphics[width=1\linewidth]{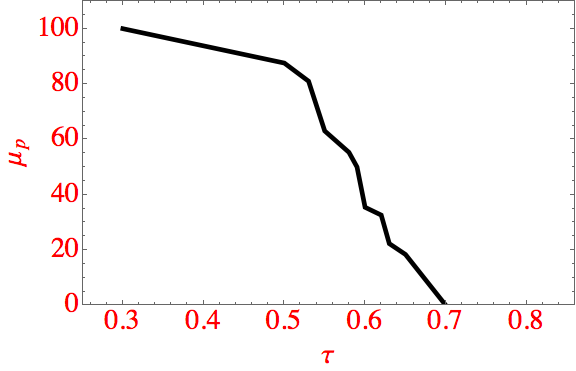}
  \caption{Polarization mean}
  \label{fig:NormalPolPeak}
\end{subfigure}%
\hspace{.6in}
\begin{subfigure}{.35\textwidth}
  \includegraphics[width=1\linewidth]{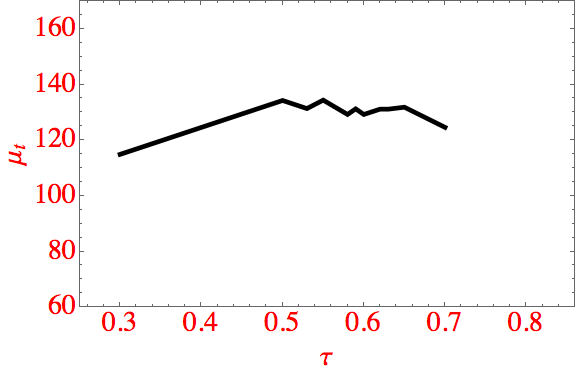}
  \caption{Mean of stabilization time}
  \label{fig:NormalTimePeak}
\end{subfigure}%
\caption{Effect of $\tau$ on normal initials}
\label{fig:NormalPeak}
\end{figure}

\begin{table}[!htb]
\caption{Normal initial opinions with different tent potentials with $\tau \in [0.3,0.7]$.}
\begin{center}
    \begin{tabular}{| l | l | l | l | l | l| l| l| l | l | l | p{1cm} |}
    \hline
    $\tau$  & .3 & .5  & .53 &  .55 & .58 & .59 & .6 & .62 & .63 & .65 & .7  \\ \hline
  $\mu_p$ & 100 & 87.57 & 81.02 & 63.04 & 55.31 & 50.02 & 35.45 & 32.63 & 22.32 & 18.41 & 0  \\ \hline
  $\mu_s$  & 114.91 & 134.27 & 131.36 & 134.39 & 129.25 & 131.26 & 129.24 & 131.15 & 131.14 & 131.85 & 124.61  \\ \hline
    \end{tabular}
\end{center}
 \label{table:normalPotTable}
\end{table}

In Figure~\ref{fig:PeakVsAll} it is clear that as the peak of the tent
function increases, all experiments with different samples from
different distributions show similar behavior. 

In general, one can first consider the continuous flow generated by
the vector field the potentials generate and then consider the random
walk behavior that results when we (essentially) compute the gradients
asynchronously. These two factors: geometry of state space (determined
by $\tau$) and random walk behavior (from random negotiation order),
give rise these observed results. Because the value of $\tau$
changes the volumes of the attractors of the consensus clusters and
the polarization clusters, a full explanation of these figures will
wait for the next paper.

\begin{figure}[httb!]
  \centering
  \includegraphics[width=.35\linewidth]{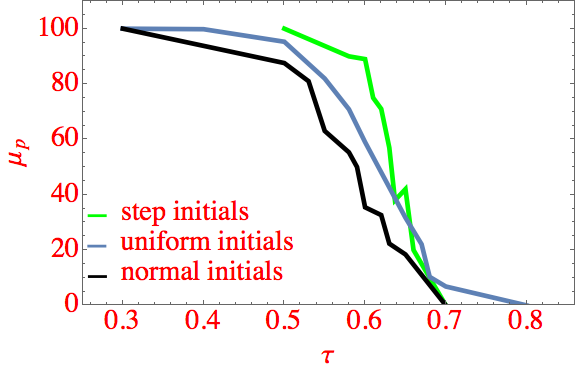}
\caption{Polarization mean vs potential $\tau$ parameter}
\label{fig:PeakVsAll}
\end{figure}

\subsection{Effect of Learning Rate}

In this section everything is fixed except the learning rate, which we
vary to see how it effects the final state of the system and the
stabilization time.  In Figure~\ref{fig:LR-effect} initial opinions and
order of negotiations/conversations are the same, only the learning
rate is varied.  As we can see, the choice of learning rate can cause the
long term equilibrium state to change from a single consensus cluster to
two polarized clusters.  There remains work to be performed to understand
how sensitive an initial opinion state is to such effects as a function
of the potential, learning rate, and distribution of initial opinion
values.

\begin{figure}[httb!]
  \centering
\begin{subfigure}{.35\textwidth}
  \includegraphics[width=1\linewidth]{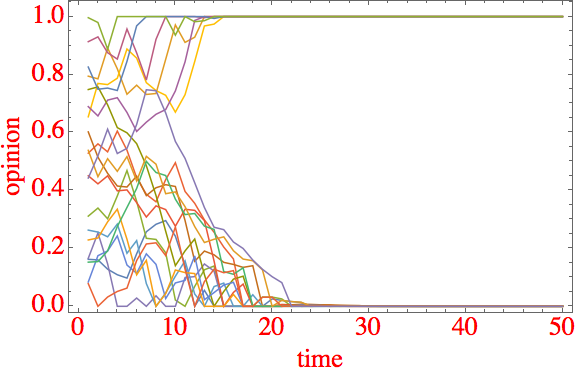}
  \caption{$\alpha = 0.2$ }
  \label{fig:LRS}
\end{subfigure}%
\hspace{.5in}
\begin{subfigure}{.35\textwidth}
  \includegraphics[width=1\linewidth]{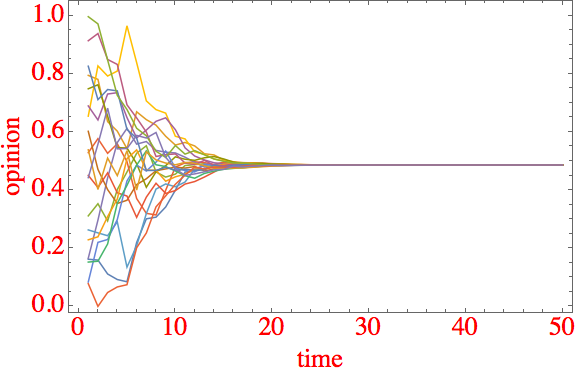}
  \caption{$\alpha = 0.3$ }
  \label{fig:LRL}
\end{subfigure}%
\caption{Learning rate effect}
\label{fig:LR-effect}
\end{figure}

\subsection{Coupling Experiment}
In this section we take a very brief peek at the effect of coupling
strength in a system with two topics and a fully connected network of
20 agents. In these 6 experiments -- one each for 6 different coupling
strengths, each time step is comprised of 10 disjoint (but random)
interactions. The same set of interactions is used for each coupling
strength. The experiments are terminated after 1000 time steps.

Without coupling, topic 1 polarizes and topic 2 converges to
consensus.  Introducing a small amount of (symmetric) coupling between
the topics for all individuals, we see that the consensus state for
topic 2 is destabilized.  It is worth noting that there exist a small number
of individuals that fail to reach a polarized state when coupling occurs.
The exact cause of this behavior is the subject of future work, but we
hypothesize that this is due to a small number of individuals starting
with an initial opinion state with the following properties.  For topic 1,
they start with an opinion that leads to the consensus state with a set of
other individuals.  On topic 2, since everyone moves to consensus there exists
only one cluster in equilibrium.  When coupling is enabled, the topic 2
consensus state is perturbed and polarization occurs.  In the case of the
individuals who oscillate over the long term, they likely are in a situation
where their opinion on topic 2 polarizes with the \emph{opposite} set of
individuals than they cluster with on topic 1.  As such, they experience a
constant tug from each population towards the opposing polarized state and
cannot reach and stay at the polarized opinion state for either topic.

%  exist in a state where they
% hold an opinion that polarizes into a subpopulation, but they share
% the opposite opinion of the other topic with their peers in that
% cohort. So they are constantly torn: on the one hand, they want to
% join their polarized cohort, but they also really want to hold an
% opposing opinion about the other topic.  As coupling gets stronger,
% interestingly, we can see that they cause some noise in the polarized
% populations too.
% Polarization also takes longer to occur, since instead
% of pulling immediately to polarization, the existence of the consensus
% opinion causes individuals to take longer to pull both topics together
% internally. 

% \tcbox{\includegraphics[width=5cm]{T1-coeff-1.png}}

\begin{figure}[htp]
%\vspace{0.25cm}
\begin{adjustbox}{varwidth=\textwidth,fbox,center}
\vspace{0.25cm}
\begin{subfigure}{.32\textwidth}
  \vspace{0.25cm}
  \includegraphics[width=1\linewidth]{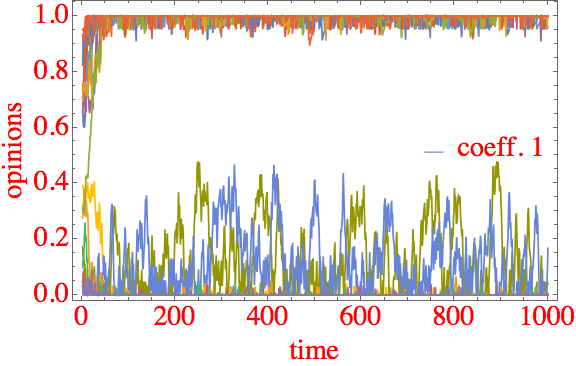}
  \label{fig:NormalPolPeak2}
\end{subfigure}%
\begin{subfigure}{.32\textwidth}
  \vspace{0.25cm}
  \includegraphics[width=1\linewidth]{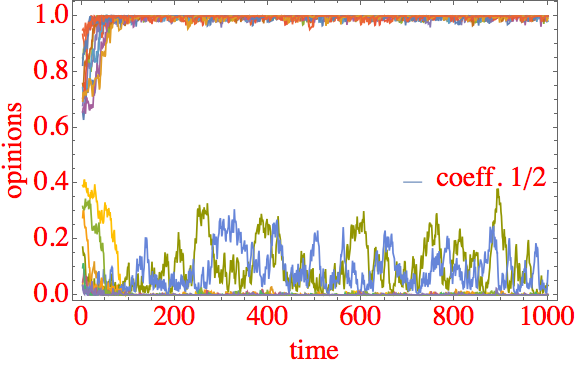}
  \label{fig:NormalPolPeak3}
\end{subfigure}%
\begin{subfigure}{.32\textwidth}
  \vspace{0.25cm}
  \includegraphics[width=1\linewidth]{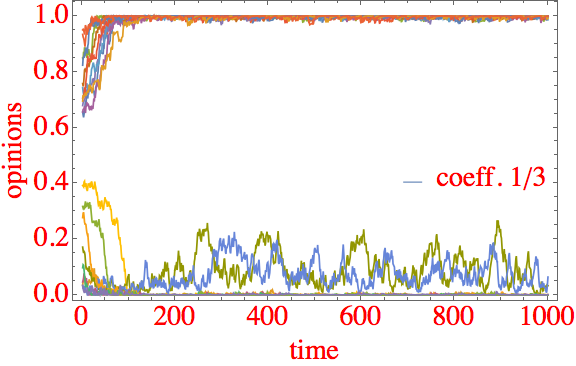}
  \label{fig:NormalPolPeak4}
\end{subfigure}%\\

\begin{subfigure}{.32\textwidth}
  \includegraphics[width=1\linewidth]{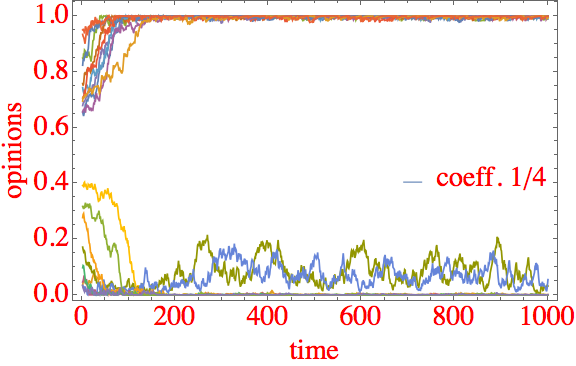}
  \label{fig:NormalPolPeak5}
\end{subfigure}%
\begin{subfigure}{.32\textwidth}
  \includegraphics[width=1\linewidth]{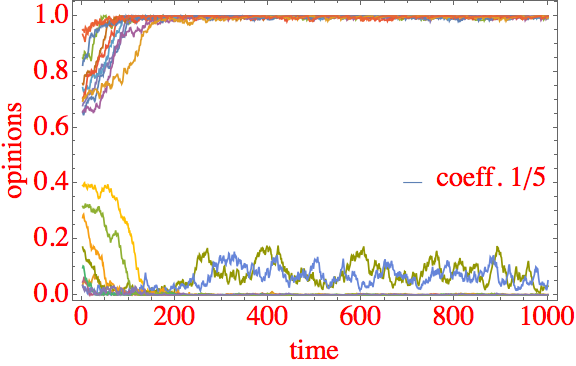}
  \label{fig:NormalPolPeak6}
\end{subfigure}%
\begin{subfigure}{.32\textwidth}
  \includegraphics[width=1\linewidth]{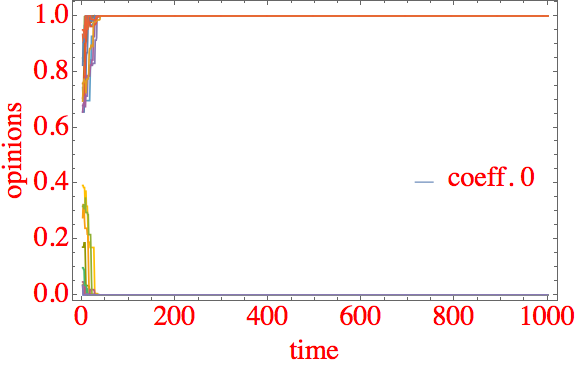}
  \label{fig:NormalPolPeak7}
\end{subfigure}%
\end{adjustbox}
\caption{Different coupling coefficients - topic 1}
\label{topic1Coup}
%\vspace*{-0.65cm}
\end{figure}
%%%%%%%%%%%%%%%%%%%%%%%%%%%%%%%%%
%%%%%%%%%%%%%%%%%%%%%%%%%%%%%%%%%
\begin{figure}[htp]
\begin{adjustbox}{varwidth=\textwidth,fbox,center}
\begin{subfigure}{.32\textwidth}
  \vspace{0.25cm}
  \includegraphics[width=1\linewidth]{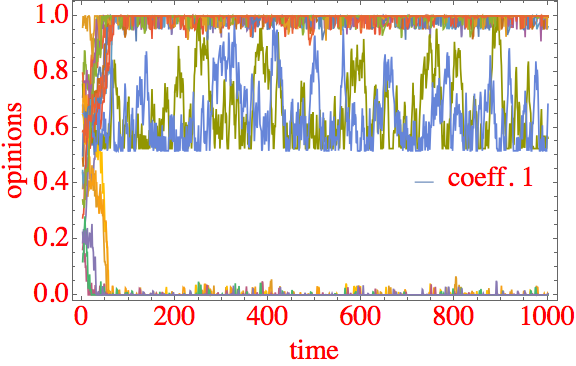}
  \label{fig:NormalPolPeak8}
\end{subfigure}%
\begin{subfigure}{.32\textwidth}
  \vspace{0.25cm}
  \includegraphics[width=1\linewidth]{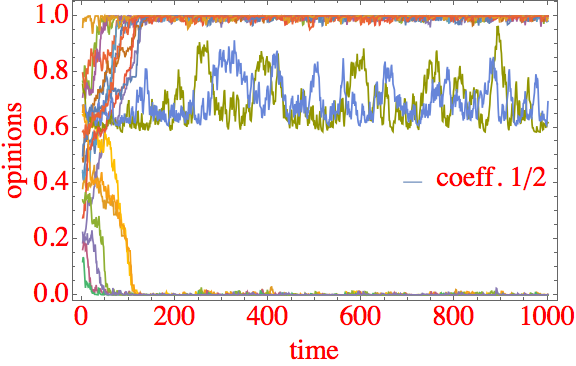}
  \label{fig:NormalPolPeak9}
\end{subfigure}%
\begin{subfigure}{.32\textwidth}
  \vspace{0.25cm}
  \includegraphics[width=1\linewidth]{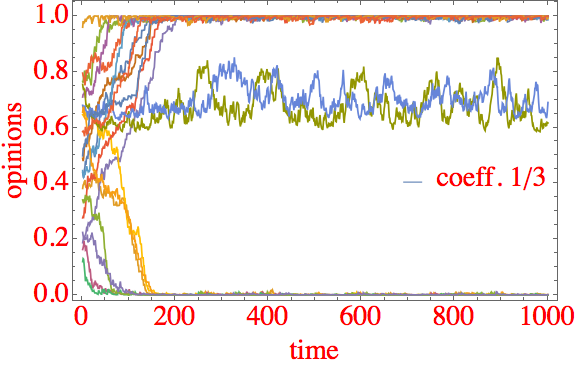}
  \label{fig:NormalPolPeak10}
\end{subfigure}%\\

\begin{subfigure}{.32\textwidth}
  \includegraphics[width=1\linewidth]{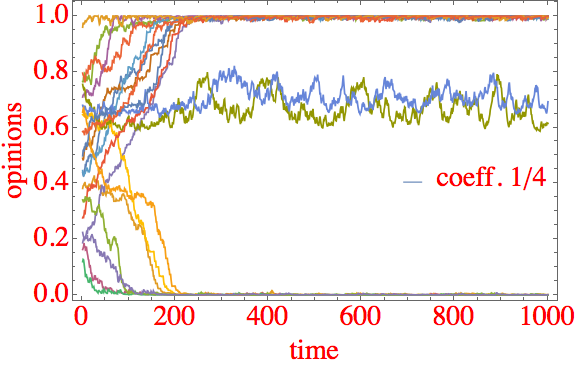}
  \label{fig:NormalPolPeak11}
\end{subfigure}%
\begin{subfigure}{.32\textwidth}
  \includegraphics[width=1\linewidth]{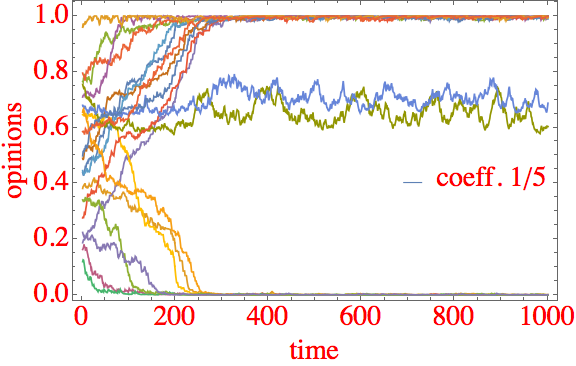}
  \label{fig:NormalPolPeak12}
\end{subfigure}%
\begin{subfigure}{.32\textwidth}
  \includegraphics[width=1\linewidth]{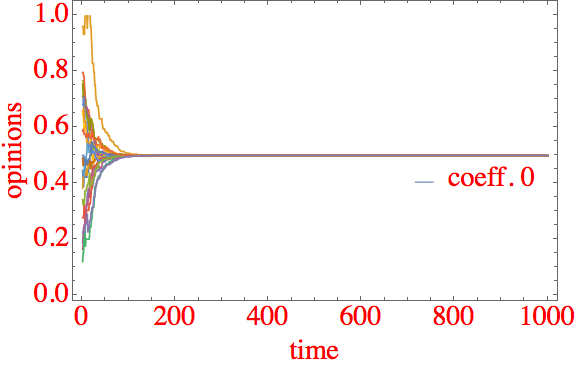}
  \label{fig:NormalPolPeak13}
\end{subfigure}%
\end{adjustbox}
\caption{Different coupling coefficients - topic 2}
\label{topic2Coup}
\end{figure}

\section{Conclusions and future work}

The work presented in this paper introduces a potential based model
with inter-topic coupling.  As defined, the model is relatively general as
we show by embedding other models from the literature in our framework.
There remain a large number of questions that can be studied based on this
work.  We will provide a few noteworthy questions that we identified while
performing this research.  First, a number of recent publications have studied
the impact of network properties (e.g., centrality measures) on dynamical
systems that are dependent on the network.  Furthermore, it is known that
real-world systems that are best represented by a network are often not static
and instead exhibit time-varying properties - both in connectivity as well
as parameters such as edge weights.

As mentioned in places within the paper, we have chosen specific
update criteria for modeling single agent pairs interacting and
sharing opinions.  Other update methods can be created that impose
differing levels of synchronous behavior, differing subpopulations,
and so on.  These will likely have an effect on long term behavior of
the model and should be studied.  One of the more interesting initial
results that we are studying in more detail is the cross-over point
shown in Fig.~(\ref{fig:UniformPolPeak}).  Why do we observe the 50/50
polarization/consensus split occur when the tent $\tau$ parameter is
approximately 0.6?  Can this value be predicted analytically from the
update rules, potential function, and model parameters?

A number of questions can be posed about the long term evolution of
the model.  We can study more about dynamics of the continuous system
such as determining the conditions under which consensus/polarization
occurs. Determining a lower bound on convergence time would be another
interesting question.  In particular, we would like to understand the
long term state reached when individuals oscillate due to coupling of
opinions that disagree with larger subpopulations - does this
oscillation run indefinitely, or does it damp out and eventually move
to a polarized state?  Many of parameters of the models can also
be made dynamic.  For example, we may allow coupling or inter-agent
weights to deviate from some equilibrium state to represent transient
phenomena (e.g., conflicts).

We believe that this work represents a noteworthy accomplishment in
opinion dynamics research for two reasons: it provides a flexible framework
for exploring variant models based on a common core, and by adopting a common
framework we are able to then use theorems about one model to reason about
others.  The common model framework will allow such translation of theorems
and properties from one model to another, shedding light on models that would
be difficult to analyze directly.

% \section{Clarifications}

% Here Kevin was trying to say the model in which at a given time step each agent is influenced by all of her neighbors.
% {\color{red}{Kevin: May be we can simultaneously calculate the increments and add them together ...}}

% The reason that in Coupling we have $x$ and $y$ in $\phi(|d|, x , y)$ is that we can have different $\phi$'s for every pair of topics! However in our experiment we are going to use only one function for all topics.

% The reason that we have $i < j$ is that if the influence of $i$ on $j$ and $j$ on $i$ are the same, then we do not have to compute it twice! 
% {\color{blue}{Kevin also added this new one: Simultaneously minimize.
% $$\sum_{i<j}\int \psi(|d_{ij}|,x,y)k(x,y)dxdy$$
% }}

\section{Appendix}\label{sec:Appendix}
Here we will show how regions of fixed points look like assuming each person or each edge has a tent potential function assigned to it with $\tau = 0.5$. The regions are given by:
\begin{itemize}
\item \:\: $R_1 = \{0 < o_2 - o_1 < \tau\} \cap \{0 < o_3 - o_2 < \tau\}  \cap \{ \tau < o_3 - o_1 \leq 1 \}$, where $o_1 < o_2 < o_3$.

\item \:\: $R_2 = \{0 < o_1 - o_2 < \tau\} \cap \{0 < o_3 - o_1 < \tau\}  \cap \{ \tau < o_3 - o_2 \leq 1 \}$, where $o_2 < o_1 < o_3$.

\item \:\: $R_3 = \{0 < o_3 - o_1 < \tau\} \cap \{0 < o_2 - o_3 < \tau\}  \cap \{ \tau < o_2 - o_1 \leq 1 \}$, where $o_1 < o_3 < o_2$.

\item \:\: $R_4 = \{0 < o_3 - o_2 < \tau\} \cap \{0 < o_1 - o_3 < \tau\}  \cap \{ \tau < o_1 - o_2 \leq 1 \}$, where  $o_2 < o_3 < o_1$.

\item  \:\: $R_5 = \{0 < o_1 - o_3 < \tau\} \cap \{0 < o_2 - o_1 < \tau\}  \cap \{ \tau < o_2 - o_3 \leq 1 \}$, where $o_3 < o_1 < o_2$.

\item  \:\: $R_6 = \{0 < o_2 - o_3 < \tau\} \cap \{0 < o_1 - o_2 < \tau\}  \cap \{ \tau < o_1 - o_3 \leq 1 \}$, where $o_3 < o_2 < o_1$.

\end{itemize}
Note that in the first two cases, for $R_1$ and $R_2$ ,the condition $ o_1 < o_2$ and $o_2 < o_1$ proves that $R_1 \cap R_2 =$ \O. Similarly, $R_i \cap R_j =$ \O, $\: 1 \leq i,j \leq 6$.\\
In the first case we have \[R_1 = \{0 < o_2 - o_1 < \tau \} \cap \{0 < o_3 - o_2 < \tau \} \cap \{\tau < o_3 - o_1 \leq 1 \} \]
The first part of this region would give us two planes given by 

\begin{equation}
\label{eq:plane1}
P1 \equiv o_2 - o_1 = 0 \:\: \textbf{and} \:\: P2 \equiv o_2 - o_1 = \tau
\end{equation}
This region is drawn in Fig. (\ref{fig:FirstPartofR}). And of course there are the two boundary planes, the cube walls $o_1 = 0$ and $ o_2 = 1 $, which we do not mention it, because the unit cube is the whole space we are existing in, but we would not forget about them. The planes given by Eq.~(\ref{eq:plane1}) is intersected by the two planes from the second region given by:

\begin{equation}
\label{plane2}
P3 \equiv o_3 - o_2 = 0 \:\: \textbf{and} \:\: P4 \equiv o_3 - o_2 = \tau
\end{equation} 

drawn in Fig.~(\ref{fig:SecondPartofR}).

\begin{figure}[httb!]
  \centering
\begin{subfigure}{.35\textwidth}
  \includegraphics[width=1\linewidth]{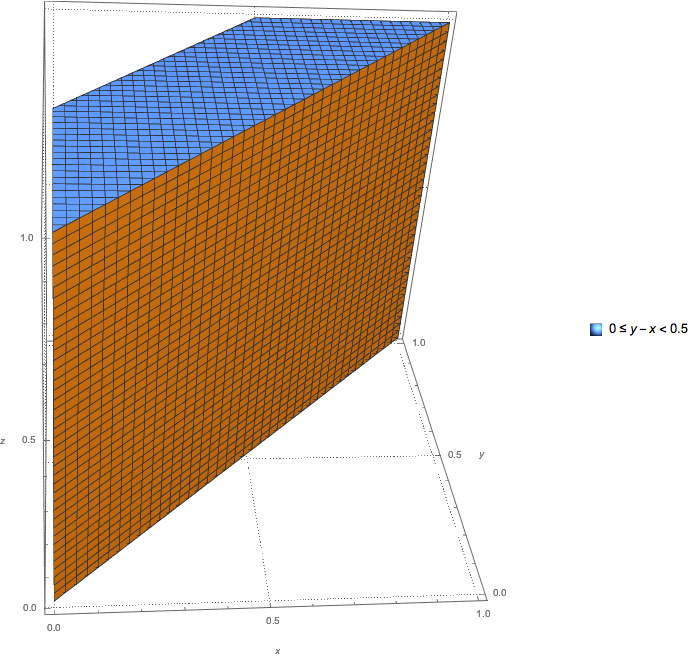}
  \caption{First Part of $R_1$}
  \label{fig:FirstPartofR}
\end{subfigure}%
\hspace{.5in}
\begin{subfigure}{.35\textwidth}
  \includegraphics[width=1\linewidth]{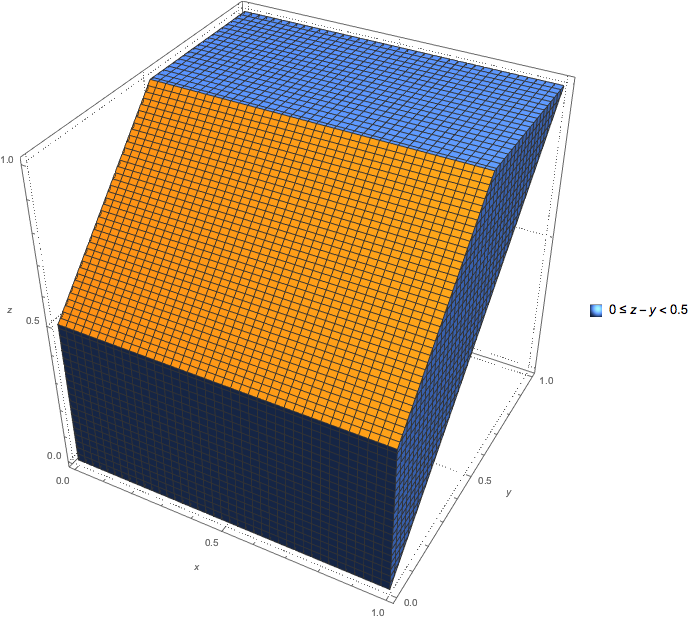}
  \caption{Second Part of $R_1$}
  \label{fig:SecondPartofR}.
\end{subfigure}%
\caption{First two parts of $R_1$}
\label{fig:FirstTwoParts}
\end{figure}
Therefore, $P1$ intersects $P3$ in a line $L_{13}$ and $P4$ in a line $L_{14}$, and the same happens for $P2$.

The normal vector of $P1$ and $P2$ is $\mathbf{n_1} = \mathbf{n_2} =\begin{bmatrix} -1 & 1 & 0 \end{bmatrix}$, and we also have $\mathbf{n_3} = \mathbf{n_4} = \begin{bmatrix} 0 & -1 & 1 \end{bmatrix}$.
Hence, the direction of $L_{13}$, $L_{14}$,  $L_{23}$ and  $L_{24}$ all are the same and is given by \[\mathbf{v} = \mathbf{n_1} \times \mathbf{n_3} =\begin{bmatrix} 1 & 1 & 1 \end{bmatrix} = \mathbf{n_1} \times \mathbf{n_4} = \mathbf{n_2} \times \mathbf{n_3} = \mathbf{n_2} \times \mathbf{n_4}\]
%The point $(0,0,0)$ is on both $P1$ and $P3$, therefore, 
The line $L_{13}$ is given by 
\begin{equation}
\label{border1}
(o_1, \: o_2, \: o_3) = t \:(1, \:1 , \: 1) 
\end{equation}

% The point $(0,0,\tau)$ is on $L_{14}$, therefore

and $L_{14}$ is given by

\begin{equation}
\label{border2}
(o_1, \: o_2, \: o_3) = (0,0, \: \tau) + t \: (1, \:1 , \: 1)
\end{equation}
% The point $( 0 , \tau, \tau)$ is on $L_{23}$ therefore, $L_{23}$ is given by 
the lines $L_{23}$ and  $L_{24}$ are given below respectively:
\begin{equation}
\label{border3}
(o_1, \: o_2, \: o_3) = (0, \: \tau, \: \tau) + t \: (1, \:1 , \: 1) 
\end{equation}
% A point on $L_{24}$ is $( 0 , \: \tau , \: 2\tau)$ and $L_{24}$ is given by
\begin{equation}
\label{border4}
(o_1, \: o_2, \: o_3) = (0,\: \tau , \: 2\tau) + t \: (1, \:1 , \: 1) 
\end{equation}
So, the area surrounded by the first two parts of $R_1$ is given by the four planes given by Eq.~(\ref{eq:plane1}) and Eq.(\ref{plane2}) , borders are given by (\ref{border1}) - (\ref{border4}) and the surfaces of the unit cube. (Fig. \ref{fig:FirstTwoPartsIntersection})

\begin{figure}[httb!]
  \centering
\includegraphics[width=.5\linewidth]{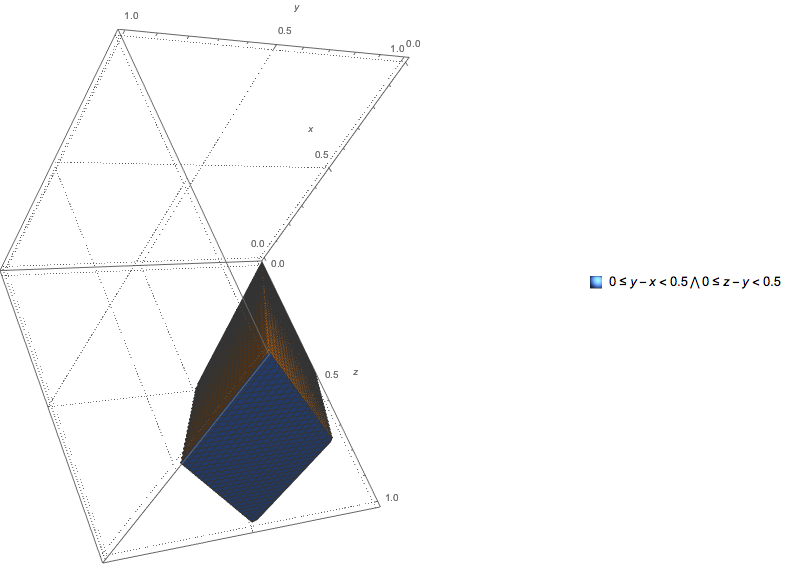}
\caption{First two parts intersection}
\label{fig:FirstTwoPartsIntersection}
\end{figure}

This area is cut by two other planes from the third part of defining $R_1$, $\{\tau < o_3 - o_1 \leq 1\}.$
\begin{equation}
\label{plane56}
P5 \equiv o_3 - o_1 = \tau \:\: \textbf{and} \:\:P6 \equiv o_3 - o_1 = 1
\end{equation}

Other than the cube's surfaces, the box has 4 sides inside the cube. Let's find the intersection of $P5$ with those 4 sides of the box. We know 
\[\mathbf{n_1} = \mathbf{n_2} = \begin{bmatrix} -1 & 1 & 0 \end{bmatrix},\: \: \mathbf{n_3} = \mathbf{n_4} = \begin{bmatrix} 0 & -1 & 1 \end{bmatrix} \: \text{and} \: \:\mathbf{n_5} =\begin{bmatrix} -1 & 0 & 1 \end{bmatrix} \]
Consequently, 
\[\mathbf{n_1} \times \mathbf{n_5} = \begin{bmatrix} 1 & 1 & 1 \end{bmatrix} \:\:\text{and }\mathbf{n_3} \times \mathbf{n_5} = \begin{bmatrix} -1 & -1 & -1 \end{bmatrix} \]
% Consequently, by choosing the points
% \[\left\{  \begin{array}{lllllll}
% p_{15} &=& (0,\: 0,\: \tau)\\
% p_{25} &=& (0,\: \tau,\: \tau)\\
% p_{35} &=& (-\tau,\: 0,\: 0)\\
% p_{45} &=& (0, \: 0,\: \tau)\\
%  \end{array} \right.\]
% where $p_{ij}$ lies on the intersection of planes ``i'' and ``j'' 
and we get:
 \[\left\{  \begin{array}{lllllll}
 L_{15} &\equiv& \:  (0,\: 0,\: \tau)    &+& t \: (1, \:1, \:1)\\
L_{25} &\equiv& \:  (0,\: \tau,\: \tau) &+& t \: (1, \:1, \:1)\\
L_{35} &\equiv& \:  (-\tau,\: 0,\: 0)   &+& t \: (1, \:1, \:1)\\
L_{45} &\equiv& \:  (0, \: 0,\: \tau)   &+& t \: (1, \:1, \:1)\\
  \end{array} \right.\]
Hence, the region is given by:
\[\left\{  \begin{array}{lllllll}
P1 \text{ from \:} L_{13} &\text{to}& L_{14}\\
P2 \text{ from \:} L_{23} &\text{to}& L_{24}\\
P3 \text{ from \:} L_{13} &\text{to}& L_{23}\\
P4 \text{ from \:} L_{14} &\text{to}& L_{24}\\
  \end{array} \right.\]
This box is cut by the fifth plane, $o_3 - o_1 = \tau $. Observe that $L_{14} = L_{15} = L_{45}$ and $L_{23} = L_{25} = L_{35}$. The new box has five sides to it of which two are the unit cube walls, and the other three are P5, P2 and P4. (Fig. \ref{fig:R1})

\begin{figure}[httb!]
  \centering
\begin{subfigure}{.4\textwidth}
  \includegraphics[width=1\linewidth]{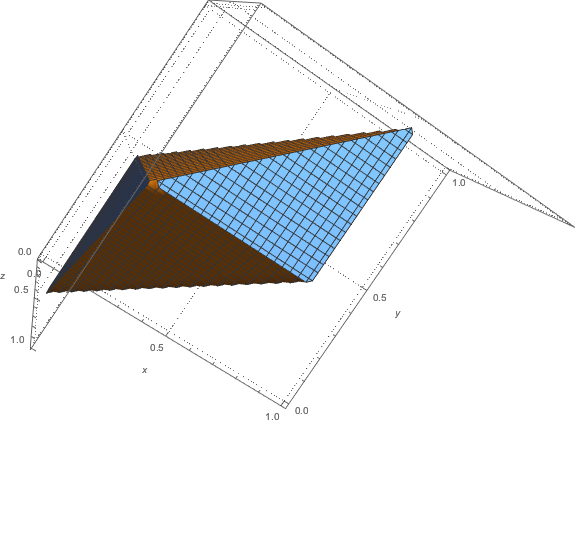}
  \caption{ $R_1$-Angle One}
  \label{fig:R1AngleOne}
\end{subfigure}%
\hspace{.5in}
\begin{subfigure}{.4\textwidth}
  \includegraphics[width=1\linewidth]{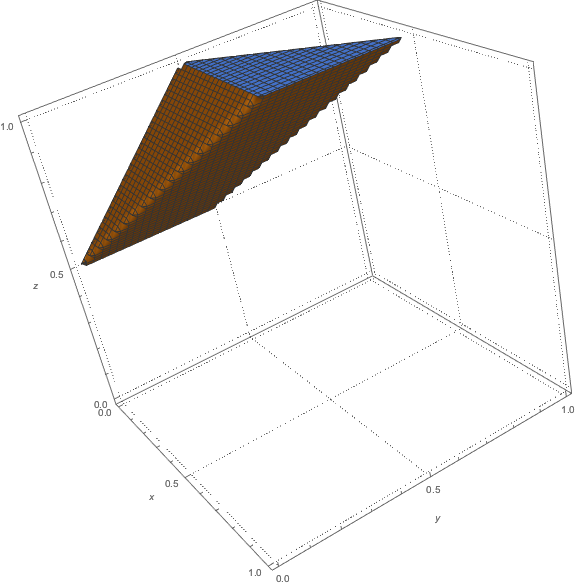}
  \caption{$R_1$-Angle Two}
  \label{fig:R1AngleTwo}
\end{subfigure}%
\caption{ The $R_1$ region}
\label{fig:R1}
\end{figure}
\pagebreak
\bibliographystyle{plain}
\bibliography{MyCollection}

\end{document}